\documentclass[tightenlines,superscriptaddress,eqsecnum,floats,nofootinbib]{revtex4}
\usepackage{amsmath,amssymb,amsthm,mathrsfs,bm,physics}
\usepackage{graphicx,color}
\usepackage[utf8]{inputenc}
\usepackage{hyperref}

\newtheorem{Theorem}{Theorem}[section]

\newtheorem{Lemma}[Theorem]{Lemma}

\def\be{\begin{eqnarray}}
\def\ee{\end{eqnarray}}

\newcommand{\cc}{\mathcal C}

\newcommand{\ch}{\mathcal H}

\newcommand{\cl}{\mathcal L}

\newcommand{\calp}{\mathcal P}

\newcommand{\cs}{\mathcal S}

\newcommand{\cv}{\mathcal V}

\newcommand{\fa}{\mathfrak{a}}  
\newcommand{\fb}{\mathfrak{b}}  
\newcommand{\fc}{\mathfrak{c}}  
\newcommand{\fd}{\mathfrak{d}}

\newcommand{\fh}{\mathfrak{h}}

\newcommand{\fv}{\mathfrak{v}}

\newcommand{\R}{\mathbb{R}}

\renewcommand{\a}{\mathfrak{a}}
\renewcommand{\b}{\mathfrak{b}}
\newcommand{\sig}{\sigma}

\newcommand{\cyl}{\mathbf{Cyl}}

\newcommand{\sgn}{\mathrm{sgn}}
\newcommand{\Pl}{\ell_p}
\newcommand{\clb}{{\underline{b}}}
\newcommand{\clc}{{\underline{c}}}
\newcommand{\clv}{{\underline{\fv}}}
\newcommand{\clz}{{\underline{\zeta}}}

\newcommand{\efh}{{H_{\rm eff}}}
\newcommand{\conb}{{\underline{\tilde{b}}}}
\newcommand{\conc}{{\underline{\tilde{c}}}}
\newcommand{\conv}{{\underline{\tilde{\fv}}}}
\newcommand{\conz}{{\underline{\tilde{\zeta}}}}

\begin{document}
\title{Reduced phase space quantization of black holes: Path integrals and effective dynamics}
\author{Cong Zhang\footnote{czhang(AT)fuw.edu.pl}}
\affiliation{Faculty of Physics, University of Warsaw, Pasteura 5, 02-093 Warsaw, Poland}
\begin{abstract}
We consider the loop quantum theory of the spherically symmetric model of gravity coupled to Gaussian dust fields, where the Gaussian dust fields provide a material reference frame of the space and time to deparameterize gravity. This theory, used to study the quantum features of the spherically symmetric black hole, is constructed based on a 1-dimensional lattice $\gamma\subset\mathbb R$. Taking advantage of the path integral formulation, we investigate the quantum dynamics and obtain an effective action. With this action, we get an effective continuous description of this quantum lattice system which is not the same as the one described by the effective Hamiltonian used in \cite{Han:2020uhb}, i.e. the classical Hamiltonian with the holonomy correction.  It turns out that the Hamiltonian derived in this paper returns that used in \cite{Han:2020uhb} only for macro black holes, since the lattice $\gamma$ is required to be sufficiently fine. Indeed, it is necessary to propose this fine-grained lattice structure in order to well describe the underlying lattice theory by the continuous description.
\end{abstract}
\maketitle

%

\section{Introduction}
Loop quantum gravity (LQG), as one of the most promising candidate for the theory of quantum gravity
\cite{ashtekar2004background,rovelli2004quantum,han2007fundamental,thiemann2008modern}, has led to several significant achievements on the quantum features of spacetime \cite{Rovelli:1993bm,rovelli1988knot,rovelli1995discreteness,ashtekar1997quantum,ashtekar1997quantumII}.
However, there are still many unsettled  issues, particularly those on the dynamics of LQG. In spite of these efforts towards the full LQG dynamics \cite{Thiemann2006,Thiemann2006a,EPRL,Domagala:2010bm,Husain:2011tk,Yang:2015zda,Alesci:2015wla,Dapor:2017rwv,Lang:2017beo,Zhang:2018wbc,Zhang:2019dgi}, one applied the loop quantization procedures to symmetry-reduced models of gravity. The resulting models (see e.g. \cite{ashtekar2003mathematical,Bojowald:2006da,Ashtekar:2006wn,Ashtekar:2008zu,Yang:2009fp,Ashtekar:2009vc,Agullo:2016tjh} for cosmological case, and see e.g.  \cite{bojowald2004spherically,Ashtekar:2005qt,Modesto:2005zm,Bohmer:2007wi,chiou2008phenomenological,Chiou:2012pg,Gambini:2013hna,corichi2016loop,Olmedo:2017lvt,Ashtekar:2018cay,Ashtekar:2018lag,BenAchour:2018khr,Bodendorfer:2019cyv,zhang2020loop,gan2020properties,rastgoo2021deformed,Gambini:2020nsf,Han:2020uhb} for black hole case), regarded as some type of symmetry-reduced sectors of LQG, yield valuable insight into the full LQG theory, regardless of the debates on embedding them into the full LQG \cite{Engle:2007zz,Bodendorfer:2015hwl,Bodendorfer:2016tky}. 

There are two approaches to study the quantum nature of the spherically symmetric black hole (BH) with loop quantum method. An approach is to loop quantize the Schwarzschild interior. 
By virtue of the homogeneity of the Schwarzschild interior, this approach leads to models of finite degrees of freedom \cite{Ashtekar:2005qt,Modesto:2005zm,Bohmer:2007wi,chiou2008phenomenological,corichi2016loop,Olmedo:2017lvt,Ashtekar:2018cay,Ashtekar:2018lag,BenAchour:2018khr,Bodendorfer:2019cyv,zhang2020loop,gan2020properties,rastgoo2021deformed}. In this frame work, the effective dynamics was investigated in details based on various quantization strategies (see e.g. \cite{Bohmer:2007wi,chiou2008phenomenological} for $\mu_0$-scheme and $\bar\mu$-scheme, and see \cite{corichi2016loop,Ashtekar:2018cay} for the modified scheme balancing the $\mu_0$-scheme and $\bar\mu$-scheme). According to these works, the BH singularity is resolved. Particularly, the singularity is replaced by quantum bounce, which connects the BH with a white whole in the framework of \cite{Bohmer:2007wi} for $\mu_0$-scheme and the framework of \cite{Ashtekar:2018cay} for a modified scheme. However, for the $\bar\mu$-scheme, according to the analysis in \cite{Bohmer:2007wi,chiou2008phenomenological}, jumping over the quantum bounce, the classical BH gives birth to a baby BH which brings forth its own baby BH. This scenario continuous, giving the extended spacetime fractal structure, until reaching eventually a Nariai type spacetime. Moreover, the quantum dynamics, in this framework, was also investigated by \cite{zhang2020loop}, which supports the existence of BH remnant by considering a general assumption on the quantization strategies. The other approach to investigate the quantum nature of BH is to consider the BH interior and exterior as a whole (see e.g. \cite{Chiou:2012pg,Gambini:2013hna,Gambini:2020nsf,Han:2020uhb}). These works consider the spherical symmetry to reduce gravity, which leads to models of (1+1)-dimension field theory on $\mathbb{R}$. An advantage of these models is that they treat both the black hole interior and exterior in a unified manner.

The current work adapts the latter approach to the spherically symmetric model of gravity coupled to Gaussian dust fields as in \cite{Han:2020uhb}. In this model, gravity is deparameterized by the Gaussian dust fields which provide a material reference frame for both time and space. Thus its dynamics is governed by a physical Hamiltonian $\mathbf{H}$, which is identical to the Hamiltonian constraint of pure gravity with the lapse function $N=1$. Classically, this dynamics gives the Lema\^itre-Tolman-Bondi spacetimes. By applying the loop quantization procedure to this model, we construct the quantum theory of this model. As the usual loop quantum theory, the states of this quantum model are constructed based on some graph $\gamma$. In this paper, $\gamma$ is chosen to be a 1-dimensional lattice, i.e. $\gamma\subset\mathbb R$, whose vertices are denoted by $v\in V(\gamma)$. Each vertex carries some quantum numbers denoted by $\zeta(v)$ and $\fv(v)$, where $\fv(v)$ is interpreted as the volume of some region containing the vertex $v$. 
To promote the physical Hamiltonian to an operator, we use the $\bar\mu$-scheme regularization strategy to get a regularized physical Hamiltonian $\mathbf{H}_\Delta$. The Hamiltonian operator $\widehat{\mathbf{H}_\Delta}$ is thus derived by quantizing $\mathbf{H}_\Delta$. To study the quantum dynamics, the path integral formulation is borrowed to calculated the transition amplitude $A(\vec\zeta_i,\vec\fv_i,\vec\zeta_f,\vec\fv_f,T)$ from the initial state $|\vec\zeta_i,\vec\fv_i\rangle$ at time $T_i=0$ to the final state $|\vec\zeta_f,\vec\fv_f\rangle$ at time $T_f=T$ \cite{EPRL,Ashtekar:2009dn,Ashtekar:2010ve,Han:2019vpw}. It turns out that this transition amplitude can be simplified in the standard form, 
\begin{equation}\label{eq:firstEq}
\begin{aligned}
A(\vec\zeta_i,\vec\fv_i,\vec\zeta_f,\vec\fv_f,T)=\int\prod_{v'\in V(\gamma)}\mathcal D[b(v')]\mathcal D[ c(v')] \mathcal D[\ln(|\zeta(v')|)]\mathcal D[\fv(v')]e^{i\frac{1}{\hbar}S(\vec \zeta,\vec \fv,\vec b,\vec c)}
\end{aligned}
\end{equation}
where $\vec{b}$ and $\vec{c}$ are variables introduced during the derivation. According to \eqref{eq:firstEq}, the classical path can be obtained by the method of stationary phase approximation, i.e. the equation $\delta S=0$. $S(\vec \zeta,\vec \fv,\vec b,\vec c)$ is thus referred to as the effective action, from which the effective Hamiltonian $H(\vec\zeta,\vec\fv,\vec c,\vec b)$ is deduced. 

We are now in a position to reconstruct an effective classical theory from the path integral formulation. This effective theory is a 1-dimensional lattice field theory whose dynamics is govern by $H(\vec\zeta,\vec\fv,\vec c,\vec b)$. An important result of our work is that $H(\vec\zeta,\vec\fv,\vec c,\vec b)$ is not the same $\mathbf{H}_\Delta$, while the former can return to the latter if the volume of each vertex $v$, i.e. $\fv(v)$, is much larger than the Planck volume. Here, it should be recalled that $\mathbf{H}_\Delta$ is the regularized Hamiltonian and its quantization gives the Hamiltonian operator $\widehat{\mathbf{H}_\Delta}$. In the work like \cite{Han:2020uhb}, one usually uses the continuous limit of $\mathbf{H}_\Delta$, denoted by $\tilde H_{\rm eff}$, to study the effective dynamics with the holonomy corrections. However, according to our analysis, this treatment is valid only
if the following inequality holds,
\begin{equation}\label{eq:inequivality0}
\frac{2GM}{\Pl}\gg \frac{4\pi}{\alpha_0}\frac{\Pl}{\delta x}
\end{equation}
where $M$ is the ADM mass of the Schwarzschild BH (the dynamics given by $\tilde H_{\rm eff}$ can recover the Schwarzschild solution in Lema\^itre coordinate according to \cite{Han:2020uhb}), $\delta x$, taking small values, is some specific coordinate length of the edges in $\gamma$ and $\alpha_0$ is some constant. Since $\delta x$ takes small values, this inequality tells us that
the effective description by $\tilde H_{\rm eff}$ is valid for macro BHs but not for small BH. In other words, the macro and small BHs behave differently once the loop quantum effects are considered. 

For a macro BH, its dynamical behavior is described by $\tilde H_{\rm eff}$ which has been studied in \cite{Han:2020uhb}. According to the results therein, the BH singularity is resolved by quantum bounce. Far away from quantum bounce, there are two asymptotic regimes sandwiching the quantum-bounce region. In one regime, the solution is semiclassical and reduces to the Schwarzschild spacetime in Lema\^itre coordinates. In the other asymptotic regime, the effective spacetime looks like the Nariai geometry which takes the metric of $\mathrm{dS}_2\times \mathbb S^2$. Moreover, in the Nariai regime, both the area of $\mathbb S^2$ and the $\mathrm{dS}$-radius of the Nariai geometry are of quantum size. Thus, our model predicts a quantum final fate of Schwarzschild BH. However, it should be emphasized again that this is credible for the case where \eqref{eq:inequivality0} hold. 

It is worth noting that the current work is only concerned about the vacuum solution which can recover the Schwarzschild BH, even though the Gaussian dusts are considered for deparameterization. For these solutions, the physical Hamiltonian, identical to the energy of the dusts field, vanishes. Thus they do not possess any evolution with respect to the dust fields. Thinking of the object evolving here as a moving particle, we can analogize these vacuum solutions to a particle at rest, regardless of $T=0$ which means that the ``time'' itself also disappears in our situation. Moreover, let us return to the original constraint theory where the deparameterization has not done. 
Then the solutions to the constraint equations are identical to the dynamical solutions derived from the deparameterized model where we regard the dynamical solutions as functions of both gravitaty and dust fields. This identity tells us that, for the cases with non-vanishing dusts, one can pick the dust fields as time to endow the physical states with a relational-evolution picture, while for vanishing-dust case, one should choose other fields to do this. From these perspectives, it makes sense to consider the vacuum solution even in the deparameterization model.

This paper is organized as follows. In Sec. \ref{sec:two}, we review briefly the classical deparameterized model of spherically symmetric gravity coupled to Gaussian dust.  Then the quantum kinematical theory of this model is introduced in Sec. \ref{sec:three}. Adapting the $\bar\mu$-scheme, Sec. \ref{sec:four} 
presents the the physical Hamiltonian operator and its the action. With these results, the quantum dynamics is studied by using the path integral formulation in Sec. \ref{sec:six}, where the transition amplitude is rewritten in the standard path integral form. Finally, the effective dynamics is studied in Sec. \ref{sec:seven}. Sec. \ref{sec:eight} summarizes  the present work and gives some outlook on it.

\section{Classical theory}\label{sec:two}

\subsection{Deparametrized gravity with Gaussian dust} \label{Deparametrized gravity with Gaussian dust}
The current work considers the model of gravity coupled to the Gaussian dusts \cite{Han:2020uhb,Kuchar:1990vy,Giesel:2012rb}. The classical phase space of this model contains the Ashtekar-Barbero canonical conjugate pairs $(A^i_\a(y),E^\a_i(y))$ for gravity and, $(T,P)$ and $(S^a,P_a)$ with $(a=1,2,3)$ for the Gaussian dusts, where $T,S^{a=1,2,3}$ will be the clock fields and define the time and space coordinates in the dust reference frame.

The dynamics of this model is encoded in the Gaussian constraint, the diffeomorphism  constraint and the Hamiltonian constraint \cite{Kuchar:1990vy,Giesel:2012rb}, where the diffeomorphism and Hamiltonian constraints read 
\be
\cc^{\mathrm{tot}}&=&P+ h,\quad h= \cc \sqrt{1+q^{\a \b} T_{, \a} T_{, \b}}-q^{\a \b} T_{, \a} \cc_{\b},\label{cc111}\\
\cc_{\a}^{\mathrm{tot}}&=&\cc_{\a}+P T_{, \a}+P_{j} S_{, \a}^{j}
\ee
with $\cc,\cc_\a$ being the purely gravitationalHamiltonian and diffeomorphism constraints. Let $(t,\sigma)$ denote the physical time and space coordinates respectively, where their  values $t(y)$ and $\sigma(y)$ at a given spacetime point 
$y$ is given by $T(y)= t$ and $S^a(y)=\sigma^a$ respectively. Then one can parameterize $(A,E)$ with respect to the dust field $T,\ S^a$ to construct the relational observables $A_a^i(t,\sigma)$ and $E_j^a(t,\sigma)$ as 
$$
\begin{aligned}
A_a^i(t,\sig)=&\left.A_\a^i(y)\frac{\partial y^\a}{\partial\sigma^a}\right|_{T(y)\equiv t,\,S^a(y)\equiv\sig^a},\\
E_j^a(t,\sig)=&\left.E_j^\a(y)\frac{\partial\sigma^a}{\partial y^\a}\right|_{T(y)\equiv t,\,S^a(y)\equiv\sig^a}.
\end{aligned}
$$
Here one should distinguish the usual coordinate index $\a=1,2,3$ and the dust coordinate index $a=1,2,3$. It turns out that these relational observables are commutes with the diffeomorphism and the Hamiltonian constraints \cite{Dittrich:2004cb,Thiemann:2004wk,Giesel:2007wn}. Thus they are Dirac observables. Moreover,  $A^j_a(t,\sig)$ and $E_j^a(t,\sig)$ satisfy the standard Poisson bracket in the dust frame: 
\be
\{A^j_b(t,\sig'),E_i^a(t,\sig)\}=\kappa \beta\delta_{i}^j\delta^a_b\delta^{3}(\sig',\sig)
\ee 
where $\beta$ is the Barbero-Immirzi parameter and $\kappa=8\pi G$.

Taking advantage of the deparametrization procedure with respect to the dust fields, we now have a natural picture of evolution $t\mapsto (A_a^i(t,\sig),E_j^a(t,\sig))$. This evolution is generated by the physical Hamiltonian ${\bf H}_0$ which is the integration of $h$ on the slice $\mathcal S$ given by $T(y)=t_0$ for some $t_0$. $\mathcal S$ is referred to as the dust space where a natural coordinate is $\mathcal S\ni s\mapsto \sigma(s)\in\mathbb R^3$ \cite{Giesel:2007wn,Giesel:2012rb}. Since $T_{,\a}=0$ on $\mathcal S$, one has
\be\label{eq:h0}
{\bf H}_0=\frac{2\pi}{\kappa}\int_\cs\dd^3\sig\, \cc(\sig),
\ee 
where $\cc(\sig)$ is expressed in terms of the Dirac observables $A^j_a(\sig):=A^j_a(t_0,\sigma),E_j^a(\sigma)=E_j^a(t_0,\sig)$. We now have a usual Hamiltonian system which is constructed by choose a spatial slice $T(y)=t$ for some $t$. However, 
By this expression, ${\bf H}_0$ is the smeared gravity Hamiltonian $\cc$ with the lapse function $N=1$.

\subsection{Spherical symmetry reduction}

We assume the dust space $\cs\simeq \R\times S^2$ and define the spherical coordinate ${\sig}=(x,\theta,\phi)$. We further reduce the reduced phase space to the phase space $\Gamma_{\rm red}$ of spherical symmetric field configurations. The spherically symmetric conjugate variables $(A^j_a,E^a_j)$ take the form \cite{Gambini:2013hna,Chiou:2012pg,Ashtekar:2005qt}
\begin{equation}\label{eq:AEexpress}
\begin{aligned}
A^j_a\tau_i\dd \sigma^a=&A_1(x)\tau_1\dd x+\frac{1}{\sqrt{2}}( A_2(x)\tau_2+ A_3(x)\tau_3)\dd\theta+\frac{1}{\sqrt{2}}( A_2(x)\tau_3-A_3(x)\tau_2)\sin(\theta)\dd\varphi+\cos(\theta)\tau_1\dd\varphi\\
E_j^a\tau^j\frac{\partial}{\partial\sigma^a}=&E^1(x)\sin(\theta)\tau_1\partial_x+\frac{1}{\sqrt{2}}( E^2(x)\tau_2+ E^3(x)\tau_3)\sin(\theta)\partial_\theta+\frac{1}{\sqrt{2}}( E^2\tau_3-E^3\tau_2)\partial_\varphi
\end{aligned}
\end{equation}
where $\tau_j=-i/2(\text{Pauli matrix})_j$.
Let $\Gamma_{\rm red}$ denote the reduced phase space of the spherically symmetric Ashtekar variables. The symplectic form $\Omega$ on $\Gamma_{\rm red}$ reads
\begin{equation}\label{eq:symplecticform}
\begin{aligned}
\Omega(\delta_1,\delta_2)=&-\frac{1}{\kappa\beta}\int\dd^3\sigma \delta_1 A_j^a(\sigma)\wedge\delta_2 E^j_a(\sigma)\\
=&-\frac{1}{2G\beta}\int\left(\delta_1A_1(x)\wedge \delta_2 E^1(x)+\delta_1 A_2(x)\wedge\delta_2 E^2(x)+\delta_1A_3(x)\wedge\delta_2 E^3(x) \right)\dd x
\end{aligned}
\end{equation}
where $\delta_1$ and $\delta_2$ are two tangent vectors in the tangent space $T_{(A,E)}\Gamma_{\rm red}$ and $\delta_1 F\wedge\delta_2 G\equiv \delta_1 F\delta_2 G-\delta_2 F\delta_1 G$. According to $\Omega$, one has the Poisson bracket 
\begin{equation}
\{A_I(x),E^J(y)\}=2G\beta\delta(x,y)\delta_I^J,\quad I,J=1,2,3.
\end{equation}
The symmetry-reduced theory is an (1+1)-dimensional field theory.

The procedure of the deparametrization in Section \ref{Deparametrized gravity with Gaussian dust} does not solve the Gauss constraint. We still need to impose Gauss constraint to $\Gamma_{\rm red}$. Substituting \eqref{eq:AEexpress} into the general expression of the Gaussian constraint, one has the expression of $G[\lambda]$, associated to a $\mathfrak{su}(2)$-valued smeared function $\lambda(x)=\lambda^i(x)\tau_i$ \cite{bojowald2004spherically,Chiou:2012pg},
\begin{equation}\label{eq:gaussian}
G[\lambda]=4\pi\int \lambda^1(x)\left[A_2(x)E^3(x)-A_3(x)E^2(x)+\partial_xE^1(x)\right]
\end{equation}
Then the gauge transformation $\cl_G$ generated by the  Gaussian constraint, i.e.,
\begin{equation}
\cl_G(F):=e^{\{F,G[\lambda]\}}=1+\sum_{n=1}^\infty\frac{1}{n!} \underbrace{\{\cdots \{F,G[\lambda]\},G[\lambda]\}\cdots\},G[\lambda]\}}_{n \text{ folds }},
\end{equation}
gives
\be
\cl_G(A_1(x))&=&A_1(x)-\kappa\beta\partial_x\lambda^1(x)\\
\cl_G(E^1(x))&=&E^1(x)\\
\cl_G(A_2(x)+iA_3(x))&=&e^{i\kappa\beta\lambda^1(x)}(A_2+i A_3)\\
\cl_G(E^2(x)+iE^3(x))&=&e^{i\kappa\beta\lambda^1(x)}(E^2+i E^3).
\ee
That is, $A_1$ transforms as a U(1) gauge field and $A_2+i A_3$ transforms as a U(1) Higgs field.  

Given a state $(A_I,E^I)\in\Gamma_{\rm red}$, one can always gauge transform it with some $\lambda$ to vanish  $E^3$. Thus we introduce a gauge fixing condition \begin{equation}
E^3(x)=0.
\end{equation}
The gauge fixed phase space then is coordinatized by $(A_I(x),E^I(x))$ with $I=1,2$ and, the Gaussian constraint \eqref{eq:gaussian} solves $A_3(x)$, leading to
\begin{equation}
A_3(x)=\frac{\partial_xE^1(x)}{E^2(x)}.
\end{equation}

We introduce the following variables\footnote{$K_x(x),E^x(x),K_\varphi(x),E^\varphi(x)$ are the same as in e.g.  \cite{Han:2020uhb,Gambini:2013hna},}
\begin{equation}
\begin{aligned}
K_x(x):=\frac{1}{2\beta}A_1(x),\ &E^x(x)=E^1(x),\\
K_\varphi(x):=\frac{1}{\sqrt{2}\beta}A_2(x),\ &E^\varphi(x)=\frac{1}{\sqrt{2}}E^2(x).
\end{aligned}
\end{equation}
The gauge-fixed reduced phase space $\calp\subset\Gamma_{\rm red}$ consist of canonical pairs $(K_x(x),E^x(x))$ and $(K_\varphi(x),E^\varphi(x))$. Restricting the symplectic form on $\calp$, we get the non-vanishing Poisson brackets
\begin{equation}
\{K_I(y),E^I(z)\}=G\delta(y-z),\ I=x,\varphi.
\end{equation}

The classical physical Hamiltonian on $\calp$ is obtained from ${\bf H}_0$ in \eqref{eq:h0} by implementing the spherical symmetry, and adding a boundary term ${\bf H}_{\rm bdy}$ 
\begin{equation}
\mathbf{H}=\frac{2\pi}{\kappa}\int \dd x\, \cc(x)+{\bf H}_{\rm bdy}
\end{equation}
where $\cc(x)$ is expressed as 
\begin{equation}\label{eq:cx}
\begin{aligned}
\cc(x)&=\frac{\sgn(E^{\varphi}{})}{\sqrt{\left|E^{x}{}\right|}}\Bigg(-\frac{E^x{}^\prime{}^2}{2E^\varphi}-8E^{x}K_x K_\varphi-2E^{\varphi}K_\varphi^2-2E^{\varphi}\Bigg)+\left(\frac{2E^{x}{}E^{x}{}^{\prime}}{\sqrt{\left|E^{x}{}\right|}|E^{\varphi}{}|}\right)^\prime
\end{aligned}
\end{equation}
The equations of motion from ${\bf H}$ is formally the same as the pure gravity dynamics with unit lapse, although the foliation is provided by the Gaussian dust.

The space of $x$ is non-compact so boundary condition is necessary. We follow \cite{Han:2020uhb} to discuss the boundary conditions and the corresponding boundary terms ${\bf H}_{\rm bdy}$. Firstly, since we are going to discuss the quantum effective modification of the Schwarzschild black hole. It is useful to introduce the asymptotic boundary condition of the Schwarzschild at $x\to\infty$. A foliation of the Schwarzschild spacetime with unit lapse is the Lema\^itre coordinate. The asymptotic boundary condition in the Lema\^itre coordinate reads
\be
E^x\sim \left(\frac32\sqrt{R_s}x\right)^{4/3},\quad&&K_x\sim \frac{R_s}{3\times 12^{1/3} \left(\sqrt{R_s}x\right)^{4/3}},\\
E^\varphi\sim \sqrt{R_2}\left(\frac{3}{2}\sqrt{R_s}x\right)^{1/3},\quad&& K_\varphi\sim-\frac{(2/3)^{1/3}\sqrt{R_s}}{\left(\sqrt{R_s}x\right)^{1/3}}.
\ee
The following boundary Hamiltonian at $x\to\infty$ turns out to be suitable for this boundary condition for canceling boundary terms of $\delta\mathbf{H}$.  
\begin{equation}
\mathbf{H}_{\rm bdy}=-\frac{4\pi}{\kappa}\left.\left(\frac{2E^{x}{}E^{x}{}^{\prime}}{\sqrt{\left|E^{x}{}\right|}|E^{\varphi}{}|}-2\sqrt{|E^x|} \right)\right|_\infty.
\end{equation}
Alternatively, instead of the above Schwarzschild asymptotic boundary condition, we may make an infrared cut-off at $x=L\gg1$ and impose the Dirichlet boundary condition 
\be\label{eq:condition1}
\delta E^x|_{x=L}=0
\ee
In this case, ${\bf H}_{\rm bdy}$ is the same except for locating at the finite boundary $x=L$:
\be
\mathbf{H}_{\rm bdy}=-\frac{4\pi}{\kappa}\left.\left(\frac{2E^{x}{}E^{x}{}^{\prime}}{\sqrt{\left|E^{x}{}\right|}|E^{\varphi}{}|}-2\sqrt{|E^x|} \right)\right|_{x=L}.
\ee
In addition, on the other side $x\to -\infty$ or $x=-L$, we impose the Neumann boundary condition 
\be\label{eq:condition2}
E^x{}'=0.
\ee
This boundary condition will be manifest in the effective dynamics of the black hole. We do not need any boundary Hamiltonian for the Neumann boundary condition. 
By using ${\bf H}_{\rm bdy}$, the Hamiltonian $\mathbf{H}$ can be simplified to the following form
\begin{equation}\label{eq:clH}
\begin{aligned}
\mathbf{H}=-\frac{2\pi}{\kappa}\int\dd x\, \fh(x)+\frac{1}{G}\left.\sqrt{|E^x|}\right|_{\rm bdy}
\end{aligned}
\end{equation}
with
\begin{equation}\label{eq:hx}
\fh(x)=\frac{\sgn(E^{\varphi}{})}{\sqrt{\left|E^{x}{}\right|}}\Bigg(\frac{E^x{}^\prime{}^2}{2E^\varphi}+8E^{x}K_x K_\varphi+2E^{\varphi}K_\varphi^2+2E^{\varphi}\Bigg).
\end{equation}

Let $\cv(x)$ denote the diffeomorphism constraint of pure gravity. Then $\cv(x)$ in the reduced model reads
\begin{equation}\label{eq:vx}
\cv(x)=K_x(x) E^x{}'(x) -E^\varphi(x) K_\varphi'(x).
\end{equation}
It can be verified that $\{\mathbf{H},\cv(x)\}=0$. Thus the time evolution by $\mathbf{H}$ has infinitely many conserved charges $\cv(x)$. It turns out that the Schwartzchild solution in Lema\^itre coordinates is obtained with setting the initial condition $\cv(x)=0$. Then $\{\mathbf{H},\cc(x)\}\propto \cv(x)$ also vanishes, which gives another conversed charges $\cc(x)$.  

\section{Quantum kinematics}\label{sec:three}

Discretize $x$-space to 1-dimensional equidistant lattice $\gamma$ whose vertices are denoted by $v$. $e(v)$ denotes the edge starting at $v$ and oriented toward the positive direction. The sets of edges and vertices in $\gamma$ are denoted by $E(\gamma)$ and $V(\gamma)$ respectively. It is worth noting that, to make $\gamma$ finite, we restrict ourselves in the interval $[-L,L]$ with $L\gg 1$ and consider the fields satisfying the boundary conditions aforementioned.

To do quantization, the fields $K_I$ and $E^I$ ($I=x,\ \varphi$) are first smeared by using $\gamma$ as follows
\begin{eqnarray}
\theta(v)=\int_{e(v)}\dd x\ K_x,\quad p(v)=E^x({\rm mid}_{e(v)}),\label{eq:xvariable}\\
\Phi(v)= K_\varphi({\rm mid}_{e(v)}),\quad \Pi(v)=\int_{e(v)}\dd x\ E^\varphi,\label{eq:phivariable}
\end{eqnarray}
where ${\rm mid}_{e(v)}$ is the middle point within $e(v)$. Here, $K_x(x)$ is gauge fixed from the $U(1)$ gauge field and, thus, $e^{i\theta(v)}$ is its holonomy along $e(v)$. However, 
$K_\varphi$ is gauge fixed from the $U(1)$ Higgs field, and, thus, $e^{i\Phi(v)}$ is its point holonomy at $v$. 

With the variables $\theta(v)$ and $\Phi(v)$, the Hilbert space, constructed by the polymer quantization procedure, reads
\begin{equation}
\ch=\overline{\Big(L^2(\mathbb R_b,\dd\mu_h)\otimes L^2(\mathbb R_b,\dd\mu_h) \Big)^{|V(\gamma)|}}. 
\end{equation} 
where $\mathbb R_b$ is the Bohr compactification of the real line $\mathbb R$ and $\dd\mu_h$ is its Haar measure. To represent the elements in $\ch$ explicitly, let us define the space $\cyl$ of almost periodic functions of $\vec\theta=\{\theta(v)\}_{v\in V(\gamma)}$ and $\vec\Phi=\{\Phi(v)\}_{v\in V(\gamma)}$. In other words, elements in $\cyl$ are finite linear combination of the functions $T_{\vec\mu,\vec\lambda}$
\begin{equation}
T_{\vec\mu,\vec\lambda}(\vec\theta,\vec\Phi)=\prod_{v\in V(\gamma)}e^{i\lambda(v)\theta(v)+i\mu(v)\Phi(v)}. 
\end{equation}
Equipping $\cyl$ with the inner product 
\begin{equation}
\begin{aligned}
\langle\Psi_1|\Psi_2\rangle=\lim_{T\to\infty}\left(\frac{1}{2T}\right)^{2|V(\gamma)|}\int_{-T}^T\cdots\int_{-T}^T\prod_{v\in V(\gamma)}\dd\theta(v)\dd\Phi(v)\overline{\Psi_1(\vec\theta,\vec\Phi)}\Psi_2(\vec\theta,\vec\Phi).
\end{aligned}
\end{equation}
we have
\begin{equation}\label{eq:innerTmunu}
\langle T_{\{\vec\lambda,\vec\mu\}}|T_{\{\vec\lambda',\vec\mu'\}}\rangle=\delta_{\vec\lambda,\vec\lambda'}\delta_{\vec\mu,\vec\mu'}.
\end{equation}
where the right hand side is the Kronecker delta.
The Hilbert space $\ch$ is identical to the Cauchy completion of the space $\cyl$ with respect this inner product.
By using the Dirac bra-ket notions, we will denote
\begin{equation}
T_{\{\vec\lambda,\vec\mu\}}\equiv |\vec\lambda,\vec\mu\rangle.
\end{equation}
\eqref{eq:innerTmunu} means that $|\vec\lambda,\vec\nu\rangle$ forms a family of orthonormal basis of $\ch$.

In $\ch$, there are the basic operators $\widehat{e^{i\lambda_o(v)\theta(v)}}$, $\widehat{e^{i\mu_o(v)\Phi(v)}}$, $\hat p(v)$ and $\hat \Pi(v)$ with constants $\lambda_o(v)$ and $\mu_o(v)$:
\be\label{eq:basicoperators}
\widehat{e^{i\lambda_o(v)\theta(v)}}|\vec\lambda,\vec\mu\rangle&=&|\vec\lambda+\lambda_o(v)\vec\delta_v,\vec\mu\rangle\\
\widehat{e^{i\lambda_o(v)\Phi(v)}}|\vec\lambda,\vec\mu\rangle&=&|\vec\lambda,\vec\mu+\mu_o(v)\vec\delta_v\rangle\\
\hat p(v)|\vec\lambda,\vec\mu\rangle&=&\Pl^2\lambda(v)|\vec\lambda,\vec\mu\rangle,\\
\hat \Pi(v)|\vec\lambda,\vec\mu\rangle&=&\Pl^2\mu(v)|\vec\lambda,\vec\mu\rangle.
\ee
where $\Pl^2=\kappa \hbar$ and $\vec\delta_v$ is defined by
\begin{equation}
\vec\delta_v(v')=\left\{
\begin{aligned}
0,&\quad v'\neq v,\\
1,&\quad v'=v.
\end{aligned}
\right.
\end{equation}
It is remarkable that that operators $\widehat{e^{i\lambda_o(v)\theta(v)}}$ and $\widehat{e^{i\mu_o(v)\Phi(v)}}$ are not strongly continuous. Therefore, there is no operators $\hat \theta(v)$ and $\hat \Phi(v)$ such that $\widehat{e^{i\lambda_o(v)\theta(v)}}=\exp(i\lambda(v)\hat\theta(v))$ and $\widehat{e^{i\lambda_o(v)\Phi(v)}}=\exp(i\lambda(v)\hat\Phi(v))$ in this model.

Let $V_v$ denoted the volume of the region $e(v)\times\mathbb S^2$. By definition, $V_v$ is
\begin{equation}
V_v= \int_{e(v)}\dd x\int_0^\pi\dd\theta\int_0^{2\pi}\dd\varphi \left|E^{\varphi}\right|\sqrt{\left|E^{x}\right|}\sin(\theta)=4\pi|\Pi(v)|\sqrt{|p(v)|},
\end{equation}
where we used that $\sgn(E^\varphi)$ is a constant. Thus its corresponding operator $\hat{V}_v=4\pi|\hat{\Pi}(v)|\sqrt{|\hat{p}(v)|}$ takes the basis $|\vec{\lambda},\vec{\mu}\rangle$ as its eigenstate, i.e.
\begin{eqnarray}
\hat{V}_v|\vec{\lambda},\vec{\mu}\rangle=4\pi \Pl^3|\mu(v)|\sqrt{|\lambda(v)|}\,|\vec{\lambda},\vec{\mu}\rangle.
\end{eqnarray}
For convenience $\hat{V}_v$ as well as $V_v$ will be referred to as the volume (operator) at $v$.

\section{The physical Hamiltonian operator}\label{sec:four}
\subsection{$\bar\mu$-scheme regularization and quantization}
Due to the absence of the operators $\hat \theta(v)$ and $\hat\Phi(v)$, the classical expression \eqref{eq:hx} of the physical Hamiltonian needs to be regularized in terms of $e^{i\lambda(v)\theta(v)}$ and $e^{i\mu(v)\theta(v)}$. One can refer to \cite{Chiou:2012pg,Han:2020uhb} and references therein for the detailed derivations on the regularization. Generally speaking, the regularization procedure results in a family of regularized expressions $\mathbf H_{\vec\lambda,\vec\mu}$, depending on parameters $\vec\lambda$ and $\vec\mu$.  Classically, it has 
$$\lim_{\vec\lambda,\ \vec\mu\to 0}\mathbf H_{\vec\lambda,\vec\mu}=\mathbf{H}.$$
However, in the quantum theory the limit $\lim\limits_{\vec\lambda,\ \vec\mu\to 0}\hat{\mathbf{H}}_{\vec\lambda,\vec\mu}$ fails to exist, where $\hat{\mathbf{H}}_{\vec\lambda,\vec\mu}$ denotes the quantization operator of $\mathbf{H}_{\vec\lambda,\vec\mu}$.
To overcome the failure of the limit to exist, the idea proposed in loop quantum cosmology model \cite{ashtekar2003mathematical,Ashtekar:2006wn} is borrowed, which lead to the operator $\hat{\mathbf{H}}$, i.e. the quantization operator of $\mathbf{H}$, defined as  
$$\hat{\mathbf{H}}=\lim_{\vec\lambda\to\vec{\underline\lambda},\vec\mu\to\vec{\underline\mu}}\hat{\mathbf{H}}_{\vec\lambda,\vec\mu}$$ 
with some non-vanishing quantities $\vec{\underline\lambda}=\{\underline\lambda(v)\}$ and $\vec{\underline\mu}=\{\underline\mu(v)\}$. There are several strategies to  choose $\vec{\underline\lambda}$ and $\vec{\underline\mu}$. The current work focus only on the $\bar\mu$-scheme as in \cite{Chiou:2012pg,Han:2020uhb} where $\vec{\underline\lambda}$ and $\vec{\underline\mu}$ are
\begin{equation}
\underline\lambda(v)=\bar\lambda(v):=\frac{2\beta \sqrt{\Delta}\Pl\sqrt{\left|p(v)\right|}}{\Pi(v)},\ \underline\mu(v)=\bar\mu(v):=\frac{\beta\sqrt{\Delta}\Pl}{\sgn(p(v))\sqrt{\left|p(v)\right|}}.
\end{equation}
This strategy leads to the regularized physical Hamiltonian $\mathbf{H}_{\Delta}$, referred as the improved Hamiltonian. {Indeed, the improved Hamiltonian is equal to the effective Hamiltonian used in \cite{Han:2020uhb}, but replacing the integral over $x$ and, the variables $K_I$ as well as $E^I$ ($I=x,\varphi$), respectively, by discrete sum over $v\in V(\gamma)$ and, the corresponding smeared variables.
}
Explicitly, $\mathbf{H}_\Delta$ reads
\begin{equation}
{\mathbf H}_\Delta=-\frac{1}{G}\sum_v\fh_{\Delta,v}+\frac{1}{G}\sqrt{|p(v_b)|}
\end{equation}
with
\begin{equation}\label{eq:discreteh}
\begin{aligned}
\fh_{\Delta,v}
=&\frac{V_v \sin \left(\frac{\beta  \sqrt{\Delta } \Pl}{\sgn(p(v))\sqrt{\left| p(v)\right| }}\Phi(v)\right) \sin \left(\frac{2 \beta  \sqrt{\Delta } \Pl\sqrt{\left| p(v)\right| }}{\Pi(v)}\theta(v) \right)}{4\pi\beta ^2 \Delta \Pl^2}\\
&+\ \frac{V_v  \sin ^2\left(\frac{\beta  \sqrt{\Delta }\Pl }{\sgn(p(v))\sqrt{\left| p(v)\right| }}\Phi(v)\right)}{8\pi \beta ^2 \Delta\Pl^2 }+2\pi\frac{\Pi(v)^2}{V_v}+\frac{\pi}{2}\frac{1}{V_v} \left[p(v+1)-p(v)\right]^2
\end{aligned}
\end{equation}
where $v_b$ denotes the source of the edge $e_b$ connecting to the boundary and $v+1$ denotes the target of $e(v)$ and $v_b$ denote the boundary vertices. It is noted that
$\sum_v\fh_{\Delta,v}$ reproduce $\int\dd x\, \fh(x)$ in the limits of low curvature and lattice refinement ($e_v$ shrinks to zero size while vertices become dense).

The expression of $\fh_{\Delta,v}$ contains the inverse volume $1/V_v$ whose naive quantization $1/\hat{V}_v$ usually cause divergence at the zero eigenvalue of $\hat{V}_v$. 
Due to this difficulty, we need the trick 
\begin{equation}
\begin{aligned}
&V(\mathcal R)^{-1}=\frac{-4\times 27}{6}\left(\frac{2}{\kappa\beta}\right)^3\int_{\mathcal R} \left|\epsilon^{abc}\tr\left(\{A_a(x),V(\mathcal R)^{1/3}\}\{A_b(x),V(\mathcal R)^{1/3}\}\{A_c(x),V(\mathcal R)^{1/3}\}\right)\right|.
\end{aligned}
\end{equation}
which can be verified by applying the Thiemann's trick \cite{thiemann2008modern}. Applying this formula to our current model, one finally get, by a straightforward calculation, that 
\begin{eqnarray}\label{trick2}
\frac{1}{V_{v}}=\lim_{\lambda,\mu\to 0}\left(\frac{27}{8\pi^{2}iG^{3}\mu^{2}\lambda}\right)\mathrm{sgn}(p(v))\mathcal{Q}_{\theta,\lambda}^{(1/3)}(v) \mathcal{Q}_{\Phi,\mu}^{(1/3)}(v) \mathcal{Q}_{\Phi,\mu}^{(1/3)}(v),
\end{eqnarray}
where
\begin{eqnarray}
\mathcal{Q}_{\theta,\lambda}^{(r)}(v)&=&e^{i\lambda\theta(v)/2}\left\{ e^{-i\lambda\theta(v)/2},V_{v}^{r}\right\} -e^{-i\lambda\theta(v)/2}\left\{ e^{i\lambda\theta(v)/2},V_{v}^{r}\right\} \\
\mathcal{Q}_{\Phi,\mu}^{(r)}(v)&=&e^{i\mu\Phi(v)/2}\left\{ e^{-i\mu\Phi(v)/2},V_{v}^{r}\right\}-e^{-i\mu\Phi(v)/2}\left\{ e^{i\mu\Phi(v)/2},V_{v}^{r}\right\}.
\end{eqnarray}
Then the inverse volume operator can be quantized through replacing the Poisson brackets by the commutators. 

Because of the expression \eqref{eq:discreteh} of $\fh_{\Delta,v}$, we consider the new holonomy operators 
\begin{equation}\label{eq:newholo}
\begin{aligned}
\hat{h}_{\theta,\Delta}(v)&:=\widehat{\exp\left[i\bar\lambda(v)\theta(v)\right]}\\
\hat{h}_{\Phi,\Delta}(v)&:=\widehat{\exp\left[i\bar\mu(v)\Phi(v)\right]}.
\end{aligned}
\end{equation}
Because of the dependence of $\bar\lambda(v)$ and $\bar\mu(v)$ on $p(v)$ and $\Pi(v)$, to define these two operators, we use the idea proposed by \cite{Ashtekar:2006wn} and set 
\begin{equation}
\Phi(v)=i\Pl^2\frac{\dd}{\dd \Pi(v)},\ \theta(v)=i\Pl^2\frac{\dd}{\dd p(v)}.
\end{equation}
Then the RHSs of \eqref{eq:newholo} are written as
\begin{equation}\label{eq:hthetap}
\begin{aligned}
\hat{h}_{\theta,\Delta}(v)\psi(\vec\lambda,\vec\mu)=&\exp\left[-\frac{2\beta \sqrt{\Delta}\sqrt{\left|\lambda(v)\right|}}{\mu(v)} \frac{\dd}{\dd \lambda(v)}\right]\psi(\vec\lambda,\vec\mu)\\
\hat{h}_{\Phi,\Delta}(v)\psi(\vec\lambda,\vec\mu)=&\exp\left[-\frac{\beta\sqrt{\Delta}\Pl}{\sgn(p(v))\sqrt{|p(v)|}} \frac{\dd}{\dd \mu(v)}\right]\psi(\vec\lambda,\vec\mu)
\end{aligned}
\end{equation}
Denoting
$$s(x)=\sgn(x)\sqrt{|x|},$$
one can calculate the RHSs of \eqref{eq:hthetap} to get
\begin{equation}\label{eq:h12onpsi}
\begin{aligned}
\left(\hat{h}_{\theta,\Delta}(v)\psi\right)(\vec\lambda,\vec\mu)&=\psi(s^{-1}[s(\vec\lambda)-\beta\sqrt{\Delta}\,\mu(v)^{-1}\vec\delta_v],\vec\mu)\\
\left(\hat{h}_{\Phi,\Delta}(v)\psi\right)(\vec\lambda,\vec\mu)&=\psi(\vec \lambda,\vec\mu-\beta\sqrt{\Delta}s(\lambda(v))^{-1}\vec\delta_v)
\end{aligned}
\end{equation}
where $s(\vec\lambda):=\{s(\lambda(v))\}_{v\in V(\gamma)}$ and $s^{-1}$ denotes its inverse function\footnote{Even though the wave function $\psi(\vec\lambda,\vec\mu)$ in \eqref{eq:hthetap} is not differentiable, the resulting operators $\hat{h}_{\theta,\Delta}(v)$ and $\hat{h}_{\Phi,\Delta}(v)$ given by \eqref{eq:h12onpsi} are well defined for those non-differentiable functions. Moreover, this definition of $\hat{h}_{\theta,\Delta}(v)$ and $\hat{h}_{\Phi,\Delta}(v)$ keeps the classical algebra structure between variables. }. 

Here it should be noted that the derivation from \eqref{eq:hthetap} to \eqref{eq:h12onpsi} is quite formal, since the wave function $\psi(\vec\lambda,\vec\mu)$ in \eqref{eq:hthetap} is not differentiable. 

For convenience, we use $\vec\zeta=\frac{s(\vec\lambda)}{\beta\sqrt{\Delta}}$ instead of $\vec\lambda$ to re-label the state $|\vec\lambda,\vec\mu\rangle$. Then we have
\begin{equation}
\sgn(\hat p(v))\sqrt{|\hat p(v)|}\,|\vec\zeta,\vec\mu\rangle=\beta\sqrt{\Delta}\Pl\zeta(v)|\vec\zeta,\vec\mu\rangle,\ \ \hat \Pi(v)|\vec\zeta,\vec\mu\rangle=\Pl^2\mu(v)|\vec\zeta,\vec\mu\rangle.
\end{equation}
With the new conventions,  \eqref{eq:h12onpsi} becomes
\begin{equation}\label{eq:hsnewbasis}
\begin{aligned}
\hat{h}_{\theta,\Delta}(v)|\vec\zeta,\vec\mu\rangle&=|\vec\zeta+\mu(v)^{-1}\vec\delta_v,\vec\mu\rangle,\\
\hat{h}_{\Phi,\Delta}(v)|\vec\zeta,\vec\mu\rangle&=|\vec\zeta,\vec\mu+\zeta(v)^{-1}\vec\delta_v\rangle.
\end{aligned}
\end{equation}
Here, the derivation uses $(\hat O\psi)(\vec\zeta,\vec\mu)=\langle\vec\zeta,\vec\mu|\hat O|\psi\rangle$ for operator $\hat{O}$. Moreover, the action of the volume operator at $v$ with the new convention is
\begin{equation}
\hat V_v|\vec\zeta,\vec\mu\rangle=4\pi\Pl^3\beta\sqrt{\Delta}|\mu(v)\zeta(v)|\,|\vec\zeta,\vec\mu\rangle. 
\end{equation} 

In the framework of the $\bar\mu$-scheme, the inverse volume in \eqref{trick2} is quantized as 
\begin{equation}
\begin{aligned}
\widehat{1/V_v}&=\lim_{\substack{\lambda\to\bar\lambda(v)\\ \mu\to\bar\mu(v)} } \widehat{\left(1/V_v\right)_{\vec\lambda,\vec\mu}},\\
\left(1/V_v\right)_{\vec\lambda,\vec\mu}&=\left(\frac{27}{8\pi^{2}iG^{3}\mu^{2}\lambda}\right)\mathrm{sgn}(p(v))\mathcal{Q}_{\theta,\lambda}^{(1/3)}(v) \mathcal{Q}_{\Phi,\mu}^{(1/3)}(v) \mathcal{Q}_{\Phi,\mu}^{(1/3)}(v).
\end{aligned}
\end{equation} 
This equation leads to the operators $\hat{\mathcal Q}_{\theta,\bar\lambda}^{(r)}$ and $\hat{\mathcal Q}_{\Phi,\bar\mu}^{(r)}$ defined as
\begin{equation}
\begin{aligned}
\hat{\mathcal Q}_{\theta,\bar\lambda}^{(r)}(v) & :=  \lim_{\lambda\to\bar\lambda}\hat{\mathcal Q}_{\theta,\lambda}^{(r)}  =  \frac{1}{i\hbar}e^{i\bar\lambda\theta(v)/2}\left[ e^{-i\bar\lambda\theta(v)/2},V_{v}^{r}\right]- \frac{1}{i\hbar}e^{-i\bar\lambda\theta(v)/2}\left[ e^{i\bar\lambda\theta(v)/2},V_{v}^{r}\right],\nonumber\\
\hat{\mathcal Q}_{\Phi,\bar\mu}^{(r)}(v)&:=\lim_{\mu\to\bar\mu}\hat{\mathcal Q}_{\Phi,\mu}^{(r)}=\frac{1}{i\hbar}e^{i\bar\mu\Phi(v)/2}\left[ e^{-i\bar\mu\Phi(v)/2},V_{v}^{r}\right]-\frac{1}{i\hbar}e^{-i\bar\mu\Phi(v)/2}\left[ e^{i\bar\mu\Phi(v)/2},V_{v}^{r}\right]
\end{aligned}
\end{equation}
where the second qualities used $[\bar\lambda,V_v]=0=[\bar\mu,V_v]$. A straightforward calculation gives that
\begin{equation}
\hat{\mathcal Q}_{\theta,\bar\lambda}^{(r)}(v)=\hat{\mathcal Q}_{\Phi,\bar\mu}^{(r)}(v)\equiv \hat{\mathcal Q}^{(r)}(v)
\end{equation}
with $$\hat{\mathcal Q}^{(r)}(v)|\vec\zeta,\vec\mu\rangle=\frac{(4\pi\Pl^3\beta\sqrt{\Delta})^r}{i\hbar}\mathcal{B}(\zeta(v),\mu(v),r)|\vec\zeta,\vec\mu\rangle$$ and
\begin{equation}
\mathcal{B}(\zeta(v),\mu(v),r)=\Big|\mu(v)\zeta(v)+\frac12\Big|^r-\Big|\mu(v)\zeta(v)-\frac12\Big|^r
\end{equation}
Therefore, the inverse volume operator reads 
\begin{equation}
\begin{aligned}
&\widehat{V_v^{-1}}|\zeta(v),\mu(v)\rangle=\left(\frac{27}{4\pi\Pl^3\sqrt{\Delta}\beta}\right)\, B(\zeta(v)\mu(v))|\zeta(v),\mu(v)\rangle
\end{aligned}
\end{equation}
with
\begin{equation}
B(x)=|x|\,\Big||x+1/2|^{1/3}-|x-1/2|^{1/3}\Big|^3.
\end{equation}

Taking advantage of the blocks defined in the last section, we can now quantize the expression \eqref{eq:discreteh} to get the operator $\hat{\fh}_{\Delta,v}$
\begin{equation}\label{eq:operatorh}
\begin{aligned}
\hat{\fh}_{\Delta,v}
=&\sqrt{\hat V_v}\frac{\widehat{ \sin \left(\bar\mu(v)\Phi(v)\right)}\widehat{ \sin \left(\bar\lambda(v)\theta(v) \right)}}{4\pi\beta ^2 \Delta\Pl^2 }\sqrt{\hat V_v}\\
&+\sqrt{\hat V_v}\frac{  \widehat{\sin ^2\left(\bar\mu(v)\Phi(v)\right)}}{8\pi \beta ^2 \Delta\Pl^2 }\sqrt{\hat V_v}+2\pi\widehat{V_v^{-1}}\hat\Pi(v)^2+\frac{\pi}{2}\widehat{V_v^{-1}}\left[\hat p(v+1)-\hat p(v)\right]^2.
\end{aligned}
\end{equation}
With this expression, the physical Hamiltonian operator  $\widehat{\mathbf{H}_\Delta}$ reads
\begin{equation}
\widehat{\mathbf{H}_\Delta}=-\frac{1}{2G}\sum_v (\hat{\fh}_{\Delta,v}+\hat{\fh}_{\Delta,v}^\dagger)+\frac{1}{G}\sqrt{|\hat p(v_b)|}.
\end{equation}
Since $\hat{V}_v$ annihilates $|\vec{\lambda},\vec{\mu}\rangle$ when $\lambda(v)=0$ or $\mu(v)=0$, the operator orderings in the first and second terms of $\hat{\fh}_{\Delta,v}$ ensure that $\hat{\fh}_{\Delta,v}$ is densely defined on the entire $\ch$.

\subsection{The action of the physical Hamiltonian operator}

It is sufficient to present the action of $\hat\fh_{\Delta,v}$ instead of $\widehat{\mathbf{H}_\Delta}$ itself.
To calculate the action, we convert the trigonometric functions in \eqref{eq:operatorh} to exponentials and define 
\begin{equation}\label{eq:handh1234}
\hat\fh_{\Delta,v}=\hat{\fh}_1(v)+\hat{\fh}_2(v)+\hat{\fh}_3(v)+\hat{\fh}_4(v),
\end{equation}
with
\begin{equation}
\begin{aligned}
\hat{\fh}_1(v)&=-\sqrt{\hat V_v}\frac{\left(\hat h_{\Phi,\Delta}(v)-\hat h_{\Phi,\Delta}(v)^{-1} \right)\left(\hat h_{\theta,\Delta}(v)-\hat h_{\theta,\Delta}(v)^{-1}\right)}{16\pi\beta ^2 \Delta \Pl^2}\sqrt{\hat V_v}\\
\hat{\fh}_2(v)&=-\sqrt{\hat V_v}\frac{\left(\hat h_{\Phi,\Delta}(v)-\hat h_{\Phi,\Delta}(v)^{-1} \right)^2}{32\pi \beta ^2 \Delta \Pl^2}\sqrt{\hat V_v}\\
\hat{\fh}_3(v)&=2\pi\widehat{V_v^{-1}}\hat\Pi(v)^2\\
\hat{\fh}_4(v)&=\frac{\pi}{2}\widehat{V_v^{-1}}\left[\hat p(v+1)-\hat p(v)\right]^2.
\end{aligned}
\end{equation}

Consider the action of operators $\hat h_{\Phi,\Delta}(v)^m\hat h_{\theta,\Delta}(v)^n$ and $\hat h_{\theta,\Delta}(v)^n\hat h_{\Phi,\Delta}(v)^m$ at first. We have
\begin{equation}
\begin{aligned}
\hat h_{\Phi,\Delta}(v)^m\hat h_{\theta,\Delta}(v)^n|\zeta(v),\mu(v)\rangle=&\left|\zeta(v)+\frac{n}{\mu(v)},\frac{\mu(v)\zeta(v)+n+m}{\zeta(v)+n/\mu(v)}\right\rangle\\
\hat h_{\theta,\Delta}(v)^n\hat h_{\Phi,\Delta}(v)^m|\zeta(v),\mu(v)\rangle=&\left|\frac{\mu(v)\zeta(v)+m+n}{\mu(v)+m/\zeta(v)},\mu(v)+\frac{m}{\zeta(v)}\right\rangle.
\end{aligned}
\end{equation}
Because $\hat h_{\Phi,\Delta}(v)$ and $\hat h_{\theta,\Delta}(v)$ shift only the components of $\vec\zeta$ and $\vec\mu$ at $v$, i.e. $\zeta(v)$ and $\mu(v)$, $\hat{\fh}_1(v)$ and $\hat{\fh}_2(v)$ preserve the values of $\zeta(v')$ and $\mu(v')$ with $v'\neq v$. Thus we are only concerned on the actions of $\hat{\fh}_1(v)$ and $\hat{\fh}_2(v)$ on $|\zeta(v),\mu(v)\rangle$.  A straightforward calculation gives that
\begin{equation}\label{eq:actionh1}
\begin{aligned}
&\hat\fh_1(v)|\zeta(v),\mu(v)\rangle\\
=&-\frac{\Pl}{4\beta\sqrt{\Delta}}\Bigg(\sqrt{|\mu(v)\zeta(v)||\mu(v)\zeta(v)+2|}\left|\zeta(v)+\frac{1}{\mu(v)},\frac{\mu(v)\zeta(v)+2}{\zeta(v)+1/\mu(v)}\right\rangle\\
&-|\zeta(v)\mu(v)|\left|\zeta(v)-\frac{1}{\mu(v)},\frac{\mu(v)\zeta(v)}{\zeta(v)-1/\mu(v)}\right\rangle-|\zeta(v)\mu(v)|\left|\zeta(v)+\frac{1}{\mu(v)},\frac{\mu(v)\zeta(v)}{\zeta(v)+1/\mu(v)}\right\rangle\\
&+\sqrt{|\mu(v)\zeta(v)||\mu(v)\zeta(v)-2|}\left|\zeta(v)-\frac{1}{\mu(v)},\frac{\mu(v)\zeta(v)-2}{\zeta(v)-1/\mu(v)}\right\rangle\Bigg).
\end{aligned}
\end{equation}
Since $\hat\fh_1(v)$ is not symmetric, we also need to consider its adjoint $\hat\fh_1(v)^\dagger$, whose action reads
\begin{equation}\label{eq:actionh1d}
\begin{aligned}
&\hat\fh_1(v)^\dagger|\zeta(v),\mu(v)\rangle\\
=&-\frac{\Pl}{4\beta\sqrt{\Delta}}\Bigg(\sqrt{|\mu(v)\zeta(v)||\mu(v)\zeta(v)+2|}\left|\frac{\mu(v)\zeta(v)+2}{\mu(v)+1/\zeta(v)},\mu(v)+\frac{1}{\zeta(v)}\right\rangle\\
&-|\zeta(v)\mu(v)|\left|\frac{\mu(v)\zeta(v)}{\mu(v)+1/\zeta(v)},\mu(v)+\frac{1}{\zeta(v)}\right\rangle-|\zeta(v)\mu(v)|\left|\frac{\mu(v)\zeta(v)}{\mu(v)-1/\zeta(v)},\mu(v)-\frac{1}{\zeta(v)}\right\rangle\\
&+\sqrt{|\mu(v)\zeta(v)||\mu(v)\zeta(v)-2|}\left|\frac{\mu(v)\zeta(v)-2}{\mu(v)-1/\zeta(v)},\mu(v)-\frac{1}{\zeta(v)}\right\rangle.
\end{aligned}
\end{equation}
For the operator $\fh_2(v)$, we have
\begin{equation}\label{eq:actionh2}
\begin{aligned}
&\hat\fh_2(v)|\zeta(v),\mu(v)\rangle\\
=&-\frac{\Pl}{8\beta\sqrt{\Delta}}\Bigg(\sqrt{|\mu(v)\zeta(v)||\mu(v)\zeta(v)+2|}\left|\zeta(v),\mu(v)+\frac{2}{\zeta(v)}\right\rangle-2|\mu(v)\zeta(v)|\Big|\zeta(v),\mu(v)\Big\rangle+\\
&\sqrt{|\mu(v)\zeta(v)||\mu(v)\zeta(v)-2|}\left|\zeta(v),\mu(v)-\frac{2}{\zeta(v)}\right\rangle\Bigg)
\end{aligned}
\end{equation}
For the operators $\fh_3(v)$ and $\fh_4(v)$, the basis $|\vec\mu,\zeta\rangle$ are their eigenvectors. Their actions is
\begin{equation}
\begin{aligned}\label{eq:actionh3}
&\hat{\fh}_3(v)|\zeta(v),\mu(v)\rangle=27\Pl \frac{\mu(v)^2}{2\beta\sqrt{\Delta}} B(\zeta(v)\mu(v)) |\zeta(v),\mu(v)\rangle,
\end{aligned}
\end{equation}
and
\begin{equation}
\begin{aligned}\label{eq:actionh4}
\hat\fh_4(v)|\vec\zeta,\vec\mu\rangle=\frac{27\Pl\beta^3\sqrt{\Delta}^3}{8}B(\zeta(v)\mu(v))\left[\sgn(\zeta((v+1))\zeta(v+1)^2-\sgn(\zeta(v))\zeta(v)^2\right]^2|\vec\zeta,\vec\mu\rangle.
\end{aligned}
\end{equation}
Substituting the expression \eqref{eq:actionh1}, \eqref{eq:actionh2}, \eqref{eq:actionh3} and \eqref{eq:actionh4} into \eqref{eq:handh1234}, one finally can get the action of $\hat{\fh}_{\Delta,v}$, which will not be written again. 

It is worth noting that the quantity 
$$\fv(v):=\zeta(v)\mu(v)$$
often appears in the  RHSs of \eqref{eq:actionh1}--\eqref{eq:actionh4} as a whole. Then, these equations can be simplified in more compact forms. To do this, we will use $\vec\fv=\{\fv(v)\}_{v\in V(\gamma)}$ instead of $\vec\mu$ to re-label the basis $|\vec\zeta,\vec\mu\rangle$. However, this re-labeling can only be defined for those $|\vec\zeta,\vec\mu\rangle$ satisfying  $\zeta(v)\neq 0$ for all $v$, because this condition guarantees the solvability of  $\mu(v)$ through $\fv(v)=\zeta(v)\mu(v)$. 
Given a state $|\vec\zeta,\vec\fv\rangle$ with $\zeta(v)\neq 0,\ \forall v$, it has
\begin{equation}\label{eq:eigennewvariable}
\sgn(\hat p(v))\sqrt{|\hat p(v)|}\,|\vec\zeta,\vec\fv\rangle=\beta\sqrt{\Delta}\Pl\zeta(v)|\vec\zeta,\vec\fv\rangle,\ \ \hat V_v|\vec\zeta,\vec\fv\rangle=4\pi \beta\sqrt{\Delta}\Pl^3|\fv(v)|\,|\vec\zeta,\vec\fv\rangle
\end{equation}
According to \eqref{eq:hsnewbasis}, one has
\begin{equation}\label{eq:hsnew2basis}
\begin{aligned}
\hat{h}_{\theta,\Delta}(v)|\vec\zeta,\vec\fv\rangle&=|\vec\zeta+\zeta(v)\fv(v)^{-1}\vec\delta_v,\vec\fv+\vec\delta_v\rangle\\
\hat{h}_{\Phi,\Delta}(v)|\vec\zeta,\vec\fv\rangle&=|\vec\zeta,\vec\fv+\vec\delta_v\rangle.
\end{aligned}
\end{equation}
By this equation, the $\bar\mu$-scheme holonomies are not densely defined on $\ch$, since they cannot act on state $|\vec\zeta,\vec\fv\rangle$ when $\fv(v)= 0$. Moreover, they are not commutative:
\begin{equation}
\begin{aligned}[]
[\hat{h}_{\theta,\Delta}(v),\hat{h}_{\Phi,\Delta}(v)]|\vec\zeta,\vec\fv\rangle=|\vec\zeta+\zeta(v)(\fv(v)+1)^{-1}\vec\delta_v,\vec\fv+2\vec\delta_v\rangle-|\vec\zeta+\zeta(v)\fv(v)^{-1}\vec\delta_v,\vec\fv+2\vec\delta_v\rangle.
\end{aligned}
\end{equation}

Even though the these holonomies $\hat{h}_{\theta,\Delta}(v)$ and $\hat{h}_{\Phi,\Delta}(v)$ are not well-defined for $|\vec\zeta,\vec\fv\rangle$ with  $\fv(v)=0$, the Hamiltonian operator $\hat\fh_{\Delta,v}$ does. For $\fv(v)=0$, it has $$\hat\fh_{\Delta,v}|\vec\zeta,\vec\fv\rangle=0.$$  
For $\fv(v)\neq 0$,  the formulas \eqref{eq:actionh1}, \eqref{eq:actionh2}, \eqref{eq:actionh3} and \eqref{eq:actionh4} can be rewritten as
\begin{eqnarray}
\hat\fh_1(v)|\zeta(v),\fv(v)\rangle&=&-\frac{\Pl}{4\beta\sqrt{\Delta}}\Bigg(\sqrt{|\fv(v)||\fv(v)+2|}\left|\zeta(v)\frac{\fv(v)+1}{\fv(v)},\fv(v)+2\right\rangle\label{eq:actionh12}\\
&&-|\fv(v)|\left|\zeta(v)\frac{\fv(v)-1}{\fv(v)},\fv(v)\right\rangle-|\fv(v)|\left|\zeta(v)\frac{\fv(v)+1}{\fv(v)},\fv(v)\right\rangle\nonumber\\
&&+\sqrt{|\fv(v)||\fv(v)-2|}\left|\zeta(v)\frac{\fv(v)-1}{\fv(v)},\fv(v)-2\right\rangle\Bigg)\nonumber,
\end{eqnarray}
\begin{eqnarray}
\hat\fh_1(v)^\dagger|\zeta(v),\fv(v)\rangle&=&-\frac{\Pl}{4\beta\sqrt{\Delta}}\Bigg(\sqrt{|\fv(v)||\fv(v)+2|}\left|\zeta(v)\frac{\fv(v)+2}{\fv(v)+1},\fv(v)+2\right\rangle\label{eq:actionh12p}\\
&&-|\fv(v)|\left|\zeta(v)\frac{\fv(v)}{\fv(v)+1},\fv(v)\right\rangle-|\fv(v)|\left|\zeta(v)\frac{\fv(v)}{\fv(v)-1},\fv(v)\right\rangle\nonumber\\
&&+\sqrt{|\fv(v)||\fv(v)-2|}\left|\zeta(v)\frac{\fv(v)-2}{\fv(v)-1},\fv(v)-2\right\rangle\nonumber,
\end{eqnarray}
\begin{eqnarray}
\hat\fh_2(v)|\zeta(v),\fv(v)\rangle&=&-\frac{\Pl}{8\beta\sqrt{\Delta}}\Bigg(\sqrt{|\fv(v)||\fv(v)+2|}\left|\zeta(v),\fv(v)+2\right\rangle-\label{eq:actionh22}\\
&&2|\fv(v)|\Big|\zeta(v),\fv(v)\Big\rangle+\sqrt{|\fv(v)||\fv(v)-2|}\left|\zeta(v),\fv(v)-2\right\rangle\Bigg)\nonumber,
\end{eqnarray}
\begin{eqnarray}
\hat{\fh}_3(v)|\zeta(v),\fv(v)\rangle&=&\frac{27\Pl}{2\beta\sqrt{\Delta}} \frac{\fv(v)^2}{\zeta(v)^2} B(\fv(v)) |\zeta(v),\fv(v)\rangle\label{eq:actionh32},
\end{eqnarray}
and
\begin{eqnarray}
\hat\fh_4(v)|\vec\zeta,\vec\fv\rangle&=&\frac{27\Pl\beta^3\sqrt{\Delta}^3}{8}B(\fv(v))\Big[\sgn(\zeta((v+1))\zeta(v+1)^2-\label{eq:actionh42}\\
&&\sgn(\zeta(v))\zeta(v)^2\Big]^2|\vec\zeta,\vec\fv\rangle.\nonumber
\end{eqnarray}

As a well-defined Hamiltonian operator, $\widehat{\mathbf{H}_{\Delta}}$ should be self-adjoint. Noting that the Hilbert space $\mathcal H$ is not separable, we thus need to choose a separable Hilbert subspace $\tilde{\mathcal H}$ which is preserved by $\widehat{\mathbf{H}_{\Delta}}$.  Given $|\vec\zeta,\vec\fv\rangle\in\mathcal H$, a natural choice of the separable subspace is the one spanned by $(\widehat{\mathbf{H}_{\Delta}})^n|\vec\zeta,\vec\fv\rangle$ for all $n\in \mathbb Z_{\geq 0}$. In $\tilde{\mathcal H}$, we choose the domain $D$ of $\widehat{\mathbf{H}_{\Delta}}$ as the subspace consisting of finite linear combinations of the basis $|\vec\zeta,\vec\fv\rangle\in \tilde{\mathcal H}$. Then, let us define a self-adjoint operator in $\tilde{\mathcal H}$ as
$$\hat N=1+\sum_{v\in V(\gamma)}(\hat V_v^2+\hat p(v)^4).$$
It is easy to shown that there exist numbers $c,d\in \mathbb R$ such that
\begin{equation}
\begin{aligned}
\|\widehat{\mathbf{H}_{\Delta}}\psi\|\leq& c \|\hat N\psi\|,\ \forall \psi\in D\\
\left|\langle \widehat{\mathbf{H}_{\Delta}}\psi,\hat N\psi\rangle-\langle\hat N\psi, \widehat{\mathbf{H}_{\Delta}}\psi\rangle \right|\leq &d\|\hat N^{1/2}\psi\|^2,\ \forall \psi\in D,
\end{aligned}
\end{equation} 
where $\langle\cdot,\cdot\rangle$ denotes the inner product in $\tilde{\mathcal H}$. Then according to Theorem X.37 in \cite{ReedSimon2}, $\widehat{\mathbf{H}_{\Delta}}$ defined on $D\subset\tilde{\mathcal H}$ is essentially self-adjoint.

\section{Path integral formulation}\label{sec:six}
Let us define $\hat\fh_i=\sum_v\hat\fh_i(v)$. By \eqref{eq:actionh12}--\eqref{eq:actionh42}, one has the matrix element of $\hat\fh_i$ which reads
\begin{eqnarray}
\langle\vec\zeta_1,\vec\fv_1|\hat\fh_1|\vec\zeta_2,\vec\fv_2\rangle&=&-\frac{\Pl}{4\beta\sqrt{\Delta}}\sum_v \prod_{v'\neq v} \delta_{\zeta_1(v'),\zeta_2(v')}\delta_{\fv_1(v'),\fv_2(v')}\sqrt{|\fv_1(v)\fv_2(v)|}\times\label{eq:matrixh1}\\
&&\Big(\delta_{\zeta_1(v)(\fv_1(v)-2),\zeta_2(v)(\fv_2(v)+1)}\delta_{\fv_1(v),\fv_2(v)+2}-\delta_{\zeta_1(v)\fv_1(v),\zeta_2(v)(\fv_2(v)-1)}\delta_{\fv_1(v),\fv_2(v)}-\nonumber\\
&&\delta_{\zeta_1(v)\fv_1(v),\zeta_2(v)(\fv_2(v)+1)}\delta_{\fv_1(v),\fv_2(v)}+\delta_{\zeta_1(v)(\fv_1(v)+2),\zeta_2(v)(\fv_2(v)-1)}\delta_{\fv_1(v),\fv_2(v)-2}\Big)\nonumber,
\end{eqnarray}
\begin{eqnarray}
\langle\vec\zeta_1,\vec\fv_1|\hat\fh_1^\dagger|\vec\zeta_2,\vec\fv_2\rangle&=&-\frac{\Pl}{4\beta\sqrt{\Delta}}\sum_v\prod_{v'\neq v}\delta_{\zeta_1(v'),\zeta_2(v')}\delta_{\fv_1(v'),\fv_2(v')}\sqrt{|\fv_2(v)||\fv_1(v)|}\times\label{eq:matrixh1p}\\
&&\Big(\delta_{\zeta_1(v)(\fv_1(v)-1),\zeta_2(v)(\fv_2(v)+2)}\delta_{\fv_1(v),\fv_2(v)+2}-\delta_{\zeta_1(v)(\fv_1(v)+1),\zeta_2(v)\fv_2(v)}\delta_{\fv_1(v),\fv_2(v)}-\nonumber\\
&&\delta_{\zeta_1(v)(\fv_1(v)-1),\zeta_2(v)\fv_2(v)}\delta_{\fv_1(v),\fv_2(v)}+\delta_{\zeta_1(v)(\fv_1(v)+1),\zeta_2(v)(\fv_2(v)-2)}\delta_{\fv_1(v),\fv_2(v)-2}\Big)\nonumber,
\end{eqnarray}
\begin{eqnarray}
\langle\vec\zeta_1,\vec\fv_1|\hat\fh_2|\vec\zeta_2,\vec\fv_2\rangle&=&-\frac{\Pl}{8\beta\sqrt{\Delta}}\sum_v\prod_{v'\neq v}\delta_{\zeta_1(v'),\zeta_2(v')}\delta_{\fv_1(v'),\fv_2(v')}\sqrt{|\fv_1(v)||\fv_2(v)|}\times\label{eq:matrixh2}\\
&&\Big(\delta_{\zeta_1(v),\zeta_2(v)}\delta_{\fv_1(v),\fv_2(v)+2}-2\delta_{\zeta_1(v),\zeta_2(v)}\delta_{\fv_1(v),\fv_2(v)}+\delta_{\zeta_1(v),\zeta_2(v)}\delta_{\fv_1(v),\fv_2(v)-2}\Big)\nonumber,
\end{eqnarray}
\begin{eqnarray}
\langle\vec\zeta_1,\vec\fv_1|\hat\fh_3|\vec\zeta_2,\vec\fv_2\rangle&=&\frac{27\Pl}{2\beta\sqrt{\Delta}}\sum_{v} \frac{\fv_2(v)^2}{\zeta_2(v)^2} B(\fv_2(v)) \prod_{v'}\delta_{\zeta_1(v'),\zeta_2(v')}\delta_{\fv_1(v'),\fv_2(v')}\label{eq:matrixh3},
\end{eqnarray}
and
\begin{eqnarray}
\langle\vec\zeta_1,\vec\fv_1|\hat\fh_4|\vec\zeta_2,\vec\fv_2\rangle&=&\frac{27\Pl\beta^3\sqrt{\Delta}^3}{8}\sum_vB(\fv_2(v))\left[\sgn(\zeta_2((v+1))\zeta_2(v+1)^2-\sgn(\zeta_2(v))\zeta_2(v)^2\right]^2\label{eq:matrixh4}\\
&&\prod_{v'}\delta_{\zeta_1(v'),\zeta_2(v')}\delta_{\fv_1(v'),\fv_2(v')}\nonumber.
\end{eqnarray}
Moreover, for the boundary term
$\widehat{H_{\rm bdy}}:=\sqrt{|\hat p(v_b)},$
 one has
\begin{equation}
\langle\vec\zeta_1,\vec\fv_1|\widehat{H_{\rm bdy}}| \vec\zeta_2,\vec\fv_2\rangle=\beta\sqrt{\Delta}\Pl|\zeta_1(v_b)|\prod_{v'}\delta_{\zeta_1(v'),\zeta_2(v')}\delta_{\fv_1(v'),\fv_2(v')}
\end{equation}

In order to simplify the matrix elements of $\hat\fh_i$, consider the space of almost periodic functions and define a functional  $\mu_h$ of these functions by 
\begin{equation}
\mu_h(f)\equiv\int_{\mathbb R}\dd\mu_h(x) f(x):=\lim_{T\to\infty}\frac{1}{2T}\int_{-T}^Tf(x)\dd x
\end{equation}
Then, it can be verified that
\begin{equation}\label{eq:deltaandintegral}
\delta_{\lambda,0}=\int_{\mathbb R}\dd\mu_h(x) e^{i\lambda x},
\end{equation}
where $\delta_{\lambda,0}$ is the Kronecker delta. This formula provide an approach to rewrite the Kronecker-delta functions appearing in the expressions of the matrix elements of $\hat{\fh}_i$, so that we can simplify these expressions to get the final path integral formula. Before proceeding further, let us consider an issue on how to deal with the delta functions taking the form $\delta_{A(\fv_1(v))\,\zeta_1(v),B(\fv_2(v))\,\zeta_2(v)}$, which appears in the expressions \eqref{eq:matrixh1}--\eqref{eq:matrixh4}. The deep idea is to re-express $\delta_{A(\fv_1(v))\zeta_1(v),B(\fv_2(v))\zeta_2(v)}$ such that it takes the form $\delta_{f(\zeta_1(v))-f(\zeta_2(v)),g}$ with $g$ independent of $\zeta_1(v)$ and $\zeta_2(v)$. Then, Lemma \ref{lmm:one} stated below can be used to rewrite the transition amplitude in the standard path integral form, i.e., the form in present of the Lebesgue measure. Thus, we need the identity
\begin{equation}
\begin{aligned}
&\delta_{A(\fv_1(v))\zeta_1(v),B(\fv_2(v))\zeta_2(v)}\\
=&\left(\delta_{\sgn(\zeta_1(v)),\sgn(\zeta_2(v))}\delta_{\sgn(A(\fv_1(v))),\sgn(B(\fv_2(v)))}+\delta_{\sgn(\zeta_1(v)),-\sgn(\zeta_2(v))}\delta_{\sgn(A(\fv_1(v))),-\sgn(B(\fv_2(v)))}\right)\times \\
&\delta_{\ln(|\zeta_2(v)|)-\ln(|\zeta_1(v)|),\ln(|B(\fv_1(v))/A(\fv_2(v))|)}
\end{aligned}
\end{equation}
Taking advantage of this identity and \eqref{eq:deltaandintegral}, we finally get\footnote{here the factors $\delta_{\sgn(\cdots),\sgn(\cdots)}$ are omitted due to the following physical consideration. At first, the limit $N\to\infty$ will be taken as shown in \eqref{eq:five26}. This limit leads to $\zeta_1(v)\to \zeta_2(v)$ and $\fv_1(v)\to\fv_2(v)$.  Thus, terms with the factor $\delta_{\sgn(\zeta_1(v)),-\sgn(\zeta_2(v))}$ can be ignored, and the factor $\delta_{\sgn(\zeta_1(v)),\sgn(\zeta_2(v))}$ can be set to $1$ directly. Second, for the factor, $\delta_{\sgn(A(\fv_1(v))),\sgn(B(\fv_2(v)))}$, after taking limit $N\to \infty$, it becomes $\delta_{\sgn(A(\fv(v))),\sgn(B(\fv(v)))}$. In the present work, we consider the transition amplitude between boundary states such that the critical path, i.e., the path dominating the path integral, does not pass the small-$\fv(v)$ region. For these cases, $\delta_{\sgn(A(\fv(v))),\sgn(B(\fv(v)))}$ can be set to 1, because in these cases the paths, for which $\delta_{\sgn(A(\fv(v))),\sgn(B(\fv(v)))}\neq  1$, have only tiny contribution to the path integral. Indeed, to check this statement, one needs to use the explicit expression of the functions $A$ and $B$.}
\begin{equation}
\begin{aligned}
&\langle \vec\zeta_1,\vec\fv_1|\sum_v\frac{1}{2}(\hat\fh_{\Delta,v}+\hat\fh_{\Delta,v}^\dagger)-\widehat{H_{\rm bdy}}|\vec\zeta_2,\vec\fv_2\rangle=\int_{\mathbb R^{2|V(\gamma)|}} \prod_{v'}\dd\mu_h(b(v'))\dd\mu_h(c(v'))\times\\
&\exp\left[\sum_{v'}i[\ln(\zeta_2(v'))-\ln(\zeta_1(v'))]c(v')+i(\fv_2(v')-\fv_1(v'))b(v')\right]H(\vec\zeta_1,\vec\fv_1,\vec\zeta_2,\vec\fv_2,\vec b,\vec c)
\end{aligned}
\end{equation}
where
\begin{equation}
\begin{aligned}
 &H(\vec\zeta_1,\vec\fv_1,\vec\zeta_2,\vec\fv_2,\vec b,\vec c)=\sum_v\Bigg\{\left(-\frac{\Pl\sqrt{|\fv_1(v)\fv_2(v)|}}{8\beta\sqrt{\Delta}}\right)  \Bigg(\exp\left[i\ln\left(\left|1+\frac{1}{\fv_2(v)}\right|\right)c(v)+2ib(v)  \right]\\
 &-\exp\left[i\ln\left(\left|1-\frac{1}{\fv_2(v)}\right|\right)c(v)  \right]-\exp\left[i\ln\left(\left|1+\frac{1}{\fv_2(v)}\right|\right)c(v)  \right]+\exp\left[i\ln\left(\left|1-\frac{1}{\fv_2(v)}\right|\right)c(v)-2ib(v)  \right]\Bigg)
+\\
&\left(-\frac{\Pl \sqrt{|\fv_1(v)\fv_2(v)|}}{8\beta\sqrt{\Delta}}\right)\Bigg(\exp\left[i\ln\left(\left|1+\frac{1}{\fv_1(v)-1}\right|\right)c(v)+2ib(v)  \right]-\exp\left[i\ln\left(\left|1-\frac{1}{\fv_1(v)+1}\right|\right)c(v)  \right]\\
&-\exp\left[i\ln\left(\left|1+\frac{1}{\fv_1(v)-1}\right|\right)c(v)  \right]+\exp\left[i\ln\left(\left|1-\frac{1}{\fv_1(v)+1}\right|\right)c(v)-2ib(v)  \right]\Bigg)+\\
&\frac{\Pl  \sqrt{|\fv_1(v)\fv_2(v)|}}{2\beta\sqrt{\Delta}}\sin^2(b(v))+\frac{27\Pl \fv_2(v)^2}{2\zeta_2(v)^2\beta\sqrt{\Delta}} B(\fv_2(v))+\\
& \frac{27\Pl\beta^3\sqrt{\Delta}^3}{8}B(\fv_2(v))\left[\sgn(\zeta_2((v+1))\zeta_2(v+1)^2-\sgn(\zeta_2(v))\zeta_2(v)^2\right]^2\Bigg\}-\beta\sqrt{\Delta}\Pl\sqrt{|\zeta_1(v_b)\zeta_2(v_b)|}.
\end{aligned}
\end{equation}
Then the matrix element of $e^{i\epsilon\hat\fh}$, with $\hat{\fh}:=\sum_v\frac{1}{2}(\hat\fh_{\Delta,v}+\hat\fh_{\Delta,v}^\dagger)-\widehat{H_{\rm bdy}}$, reads
\begin{equation}
\begin{aligned}
&\langle \vec\zeta_1,\vec\fv_1|e^{i\epsilon\hat\fh}|\vec\zeta_2,\vec\fv_2\rangle=\langle \vec\zeta_1,\vec\fv_1|(1+i\epsilon\fh)|\vec\zeta_2,\vec\fv_2\rangle\\
=&\int_{\mathbb R^{2|V(\gamma)|}} \prod_{v'}\dd\mu_h(b(v'))\dd\mu_h(c(v'))\\
&\exp\left[\sum_{v'}i(\ln(|\zeta_2(v')|)-\ln(|\zeta_1(v')|))c(v')+i(\fv_2(v')-\fv_1(v'))b(v')\right]\times\\
&\left[1+i\epsilon H(\vec\zeta_1,\vec\fv_1,\vec\zeta_2,\vec\fv_2,\vec b,\vec c)\right]+O(\epsilon^2).
\end{aligned}
\end{equation}
To simplify the last equation, we claim the following theorem
\begin{Theorem}\label{thm:sumtoint}
Given $F(\{\ln(|\zeta(v)|),\fv(v)|\}_{v\in V(\gamma)})$ a Schwartz function. Then
\begin{equation}
\begin{aligned}
&\tilde F(\{\ln(|\zeta(v)|),\fv(v)|\}_{v\in V(\gamma)}):=\sum_{\{\ln(|\zeta(v)|),\fv(v)|\}_{v\in V(\gamma)}} \langle \vec\zeta_1,\vec\fv_1|(1-i\epsilon\fh)|\vec\zeta,\vec\fv\rangle F(\{\ln(|\zeta(v)|),\fv(v)|\}_{v\in V(\gamma)})\\
=&\int_{\mathbb R^{4|V(\gamma)|}}\prod_{v'}\dd b(v')\dd c(v') \frac{\dd \zeta(v')}{\zeta(v')}\dd\fv(v')\\
&\exp\left[\sum_{v'}i(\ln(|\zeta_2(v')|)-\ln(|\zeta_1(v')|))c(v')+i(\fv_2(v')-\fv_1(v'))b(v')\right]\times\\
&\Big(1-i\epsilon H(\vec\zeta_1,\vec\fv_1,\vec\zeta_2,\vec\fv_2,\vec b,\vec c)\Big)F(\{\ln(|\zeta(v)|),\fv(v)|\}_{v\in V(\gamma)})
\end{aligned}
\end{equation}
with $\dd b(v'),\ \dd c(v'),\ \dd \zeta(v'),\ \dd\fv(v')$ being the Lebesgue measure.
\end{Theorem}
This theorem can deduced straightforwardly by the following lemmas \ref{lmm:one} and \ref{lmm:two}.
\begin{Lemma}\label{lmm:one}
Given a function 
\begin{equation}
f(x,y)=\sum_{\vec\lambda\in\mathcal L}\mu(\vec\lambda,y)e^{ig(\vec\lambda) x}
\end{equation}
where $\mathcal L\subset\mathbb R^N$ is a finite lattice in $\mathbb R^N$, $g(\vec\lambda)$ is a real function of $\vec\lambda$, and $\mu(\vec\lambda,y)$  satisfies that $\mu(\vec\lambda,\cdot)$ is a Schwartz function for each $\vec\lambda$.
Then
\begin{equation}
\sum_{y}\int\dd\mu_h(x)e^{-i (y+\alpha) x}f(x,y)=\frac{1}{2\pi}\int_{-\infty}^\infty  \dd x\int_{-\infty}^\infty\dd y\, e^{-i (y+\alpha) x} f(x,y).
\end{equation}  
with $\alpha$ being a constant.
\end{Lemma}
\begin{proof}
By definition,
\begin{equation}
\frac{1}{2\pi}\int_{-\infty}^\infty\dd x\int_{-\infty}^\infty\dd y e^{-i(y+\alpha)x}f(x,y)=\frac{1}{2\pi}\int_{-\infty}^\infty\dd x\int_{-\infty}^\infty\dd y\sum_{\vec\lambda\in\mathcal L}\mu(\vec\lambda,y)e^{i(g(\vec\lambda) -y-\alpha)x}
\end{equation}
Because $\mathcal L$ is a finite lattice, one exchanges the order between the summation and the integrals to get
\begin{equation}
\frac{1}{2\pi}\int_{-\infty}^\infty\dd x\int_{-\infty}^\infty\dd y e^{-i(y+a)x}f(x,y)=\frac{1}{2\pi}\sum_{\vec\lambda\in\mathcal L}\int_{-\infty}^\infty\dd x\int_{-\infty}^\infty\dd y\,\mu(\vec\lambda,y)e^{i(g(\vec\lambda) -y-a)x}.
\end{equation}
Since $\mu(\lambda,\cdot)$ is Schwartz function for each $\lambda$, by using the inverse Fourier inversion theorem, we have
\be
&&\frac{1}{2\pi}\int_{-\infty}^\infty\dd x\int_{-\infty}^\infty\dd y\,\mu(\vec\lambda,y)e^{i(g(\vec\lambda) -y-\alpha)x}\\
&=&\frac{1}{\sqrt{2\pi}}\int_{-\infty}^\infty\dd x e^{i(g(\vec\lambda))-\alpha)x}\left(\frac{1}{\sqrt{2\pi}}\int_{-\infty}^\infty\dd y\,\mu(\vec\lambda,y)e^{-ix y}\right)\\
&=&\mu(\vec\lambda,g(\vec\lambda)-\alpha).
\ee
Thus
\begin{equation}
\frac{1}{2\pi}\int_{-\infty}^\infty\dd x\int_{-\infty}^\infty\dd y e^{-i(y+\alpha)x}f(x,y)=\sum_{\vec\lambda\in\mathcal L}\mu(\vec\lambda,g(\vec\lambda)-\alpha)
\end{equation}
where the RHS is just $\sum_y\int\dd\mu_h(x)e^{-i x y}f(x,y)$ by definition of $\mu_h(x)$. 
\end{proof}
\begin{Lemma}\label{lmm:two}
Given a function $\ch(x,y)$ taking the form
\begin{equation}
\ch(x,y_1,y_2)=\sum_{k=1}^N f_k(y_1,y_2)e^{ig_k x},
\end{equation}
where we assume that $f_k(y_1,y_2)$ are continuous and, there exist some $\alpha$ and $\beta$ such that  
\begin{equation}\label{eq:assumptionf}
\frac{f_k(y_1,y_2)}{(1+(y_1)^2)^\alpha (1+(y_2)^2)^\beta}<\infty
\end{equation} 
Let $F(y)\in \mathcal S(\mathbb R)$,  with $\mathcal S(\mathbb R)$ being the space of Schwartz functions.
 Then
\begin{equation}\label{eq:lemma2}
\begin{aligned}
&\tilde F(y_1)=\sum_{y_2}\int\dd\mu_h(x) e^{-i(y_2-y_1) x}(1+i\epsilon \ch(x,y_1,y_2) )F(y_2)\\
=&\frac{1}{2\pi}\int_{-\infty}^\infty\dd x\int_{-\infty}^\infty\dd y_2\ (1+i\epsilon \ch(x,y_1,y_2) )F(y_2)
\end{aligned}
\end{equation}
\end{Lemma}
To prove this lemma, one only needs to verify $\mu(k,y_1,\cdot)\in\mathcal S(\mathbb R)$ by \eqref{eq:assumptionf}. Then, Lemma \ref{lmm:one} can be applied to get \eqref{eq:lemma2}.

According to Theorem \ref{thm:sumtoint}, $\langle \vec\zeta_1,\vec\fv_1|e^{i\epsilon \hat \fh}|\vec\zeta_2,\vec\fv_2\rangle$ can be written as
\begin{equation}\label{eq:expmakesensedist}
\begin{aligned}
&\langle \vec\zeta_1,\vec\fv_1|e^{i\epsilon\fh}|\vec\zeta_2,\vec\fv_2\rangle\\
=&\left(\frac{1}{2\pi}\right)^{|V(\gamma)|}\int_{\mathbb R^{2|V(\gamma)|}} \prod_{v'}\dd b(v')\dd c(v')\\
&\exp\left[\sum_{v'}i(\ln(|\zeta_2(v')|)-\ln(|\zeta_1(v')|))c(v')+i(\fv_2(v')-\fv_1(v'))b(v')\right]\times\\
&\left(1+i\epsilon H(\vec\zeta_1,\vec\fv_1,\vec\zeta_2,\vec\fv_2,\vec b,\vec c)\right)+O(\epsilon^2).
\end{aligned}
\end{equation}
It is noted that the RHS is understood as a functional on the space of Schwartz functions of $\vec\zeta_1$ and $\fv_1$, as stated in Theorem \ref{thm:sumtoint}. With these formulas, let us consider the transition amplitude $A(\vec\zeta_i,\vec\fv_i,\vec\zeta_f,\vec\fv_f,T)$
\begin{equation}\label{eq:five26}
\begin{aligned}
&A(\vec\zeta_i,\vec\fv_i,\vec\zeta_f,\vec\fv_f,T)=\langle\vec\zeta_f,\vec\fv_f|e^{-i\frac{T }{-G\hbar}\hat \fh }|\vec\zeta_i,\vec\fv_i\rangle=\langle\vec\zeta_f,\vec\fv_f|\left(e^{i \frac{T}{G\hbar N} \hat\fh}\right)^N|\vec\zeta_i,\vec\fv_i\rangle\\
=&\lim_{N\to\infty}\sum_{\{\vec\zeta_n,\vec\fv_n\}_{n=1}^N}\prod_{k=0}^N\langle\vec\zeta_{k+1},\vec\fv_{k+1}|e^{i \frac{T}{\Pl^2N}\hat\fh}|\vec\zeta_k,\vec\fv_k\rangle\\
=&\int\prod_{v'}\mathcal D[b(v')]\mathcal D[ c(v')] \mathcal D[\ln(|\zeta(v')|)]\mathcal D[\fv(v')]e^{i\frac{1}{\hbar}S(\vec \zeta,\vec \fv,\vec b,\vec c)}
\end{aligned}
\end{equation}
where $\vec\zeta_0=\vec\zeta_i$, $\vec\zeta_{N+1}=\vec\zeta_f$, $\vec\fv_0=\vec\fv_i$ and $\vec\fv_{N+1}=\vec\fv_f$ and, with excluding the boundary term,
\begin{equation}
\begin{aligned}
S(\vec \zeta,\vec \fv,\vec c,\vec b)=\int_0^T\dd t\mathcal L(\vec \zeta,\vec \fv,\vec c,\vec b)
=&\int_0^T\dd t\, \left(\sum_{v}-\frac{\hbar}{\zeta(v)}\frac{\dd \zeta(v)}{\dd t}c(v)-\hbar\frac{\dd \fv(v)}{\dd t}b(v)+\frac{1}{G} H(\vec\zeta,\vec\fv,\vec c,\vec b)\right)
\end{aligned}
\end{equation}
\begin{equation}
\begin{aligned}
&H(\vec\zeta,\vec\fv,\vec c,\vec b)=\sum_v\Bigg\{\left(-\frac{\Pl |\fv(v)| }{8\beta\sqrt{\Delta}}\right)  \Bigg(\exp\left[i\ln\left(\left|1+\frac{1}{\fv(v)}\right|\right)c(v)+2ib(v)  \right]-\exp\left[i\ln\left(\left|1-\frac{1}{\fv(v)}\right|\right)c(v)  \right]\\
&-\exp\left[i\ln\left(\left|1+\frac{1}{\fv(v)}\right|\right)c(v)  \right]+\exp\left[i\ln\left(\left|1-\frac{1}{\fv(v)}\right|\right)c(v)-2ib(v)  \right]\Bigg)
+\\
&\left(-\frac{\Pl |\fv(v)|}{8\beta\sqrt{\Delta}}\right)\Bigg(\exp\left[i\ln\left(\left|1+\frac{1}{\fv(v)-1}\right|\right)c(v)+2ib(v)  \right]-\exp\left[i\ln\left(\left|1-\frac{1}{\fv(v)+1}\right|\right)c(v)  \right]\\
&-\exp\left[i\ln\left(\left|1+\frac{1}{\fv(v)-1}\right|\right)c(v)  \right]+\exp\left[i\ln\left(\left|1-\frac{1}{\fv(v)+1}\right|\right)c(v)-2ib(v)  \right]\Bigg)+\\
&\frac{ |\fv(v)|\Pl}{2\beta\sqrt{\Delta}}\sin^2(b(v))+\frac{27\Pl \fv(v)^2}{2\zeta(v)^2\beta\sqrt{\Delta}} B(\fv(v))+ \frac{27\Pl\beta^3\sqrt{\Delta}^3}{8}B(\fv(v))\left[\sgn(\zeta((v+1))\zeta(v+1)^2-\sgn(\zeta(v))\zeta(v)^2\right]^2\Bigg\}\\
&-\beta\sqrt{\Delta}\Pl|\zeta(v_b)|.
\end{aligned}
\end{equation}

\section{The effective dynamics}\label{sec:seven}
To consider the classical limit $\hbar\to 0$, we introduce the classical fields as  
\begin{equation}
\begin{aligned}
\clv(v)=4\pi\beta\sqrt{\Delta}\Pl^3\,\fv(v),& &&\clz(v)=\beta\sqrt{\Delta}\Pl\zeta(v)\\
\clb(v)=\frac{1}{4\pi\beta\sqrt{\Delta}\Pl} b(v),& &&\clc(v)=\Pl^2 c(v).
\end{aligned}
\end{equation}
The non-vanishing Poisson brackets between these fields are
\begin{equation}\label{eq:poissondis}
\{\clb(v),\clv(v)\}=G,\ \{\clc(v),\clz(v)\}=G\clz(v).
\end{equation}
In terms of the classical fields, $H$ can be simplified as
\begin{equation}\label{eq:Horigion}
\begin{aligned}
&H(\vec\clz,\vec\clv,\vec\clc,\vec\clb)=\sum_v\Bigg\{\left(-\frac{1}{16\pi\beta^2\Delta\Pl^2}\right) |\clv(v)| \times\\
&\Bigg(\cos\left[\ln\left(\left|1+\frac{4\pi\beta\sqrt{\Delta}\Pl^3}{\clv(v)}\right|\right)\frac{\clc(v)}{\Pl^2}+ 8\pi\beta\sqrt{\Delta}\Pl\clb(v)  \right]-\cos\left[\ln\left(\left|1-\frac{4\pi\beta\sqrt{\Delta}\Pl^3}{\clv(v)}\right|\right)\frac{\clc(v)}{\Pl^2}  \right]\\
&-\cos\left[\ln\left(\left|1+\frac{4\pi\beta\sqrt{\Delta}\Pl^3}{\clv(v)}\right|\right)\frac{\clc(v)}{\Pl^2}  \right]+\cos\left[\ln\left(\left|1-\frac{4\pi\beta\sqrt{\Delta}\Pl^3}{\clv(v)}\right|\right)\frac{\clc(v)}{\Pl^2}-8\pi\beta\sqrt{\Delta}\Pl b(v)  \right]\Bigg)
+\\
&\frac{|\clv(v)|}{8\pi\beta^2 \Delta\Pl^2}\sin^2(4\pi\beta\sqrt{\Delta}\Pl\clb(v))+ \frac{27\clv(v)^2}{32\pi^2\beta\sqrt{\Delta}\clz(v)^2\Pl^3} B\left[\frac{\clv(v)}{4\pi\beta\sqrt{\Delta}\Pl^3}\right]+\\
& \frac{27}{8\beta\sqrt{\Delta}\Pl^3}B\left[\frac{\clv(v)}{4\pi\beta\sqrt{\Delta}\Pl^3}\right]\left[\sgn(\clz((v+1))\clz(v+1)^2-\sgn(\clz((v+1))\clz(v)^2\right]^2\Bigg\}-\beta\sqrt{\Delta}\Pl|\zeta(v_b)|
\end{aligned}
\end{equation}
To investigate the effect of the holonomy correction but ignore the $\hbar$-order correction, we consider the limit $\Pl\to 0$ but $\beta^2\Delta\Pl^2=\text{constant}$. Then $H$ can be simplified to the form
\begin{equation}\label{eq:Hlimit1}
\begin{aligned}
&H(\vec\clz,\vec\clv,\vec\clc,\vec\clb)=\sum_v\Bigg\{\frac{1 }{4\pi \beta^2\Delta\Pl^2}|\clv(v)|\sin\left[4\pi\beta\sqrt{\Delta}\Pl\clb(v)\right]\sin\left[\frac{4\pi\beta\sqrt{\Delta}\Pl}{\clv(v)}\clc(v)+4\pi\beta\sqrt{\Delta}\Pl\clb(v)  \right]+\\
&\frac{1}{8\pi\beta^2\Delta\Pl^2}|\clv(v)|\sin^2(4\pi\beta\sqrt{\Delta}\Pl\clb(v))+\frac{|\clv(v)|}{8\pi\clz(v)^2}+ \frac{\pi}{2}\frac{1}{|\clv(v)|}\left[\sgn(\clz(v+1))\clz(v+1)^2-\sgn(\clz(v))\clz(v)^2\right]^2\Bigg\}-|\clz(v_b)|.
\end{aligned}
\end{equation}
where it is used that
\begin{equation}
\begin{aligned}
\frac{1}{\Pl^2}\ln\left(\left|1\pm \frac{4\pi\beta\sqrt{\Delta}\Pl^3}{\clv(v)}\right|\right)&\cong \pm \frac{4\pi\beta\sqrt{\Delta}\Pl}{\clv(v)}\\
\frac{1}{\Pl^3}B\left[\frac{\clv(v)}{4\pi\beta\sqrt{\Delta}\Pl^3}\right]&\cong \frac{4\pi\beta\sqrt{\Delta}}{27|\clv(v)|}.
\end{aligned}
\end{equation}
It is worth noting that the RHS of \eqref{eq:Hlimit} returns to that of \eqref{eq:discreteh} with the assignment
\begin{equation}\label{eq:relationnewold}
\begin{aligned}
&\theta(v)=\frac{\clc(v)+\clb(v)\clv(v)}{2 \sgn(\clz(v))\clz(v)^2},\ &&p(v)=\sgn(\clz(v))\clz(v)^2,\\
 &\Phi(v)=4\pi\clz(v)\clb(v),\ &&\Pi(v)=\frac{\clv(v)}{4\pi\clz(v)}. 
\end{aligned}
\end{equation}

In the following discussion, we are only concerned  about the dynamics for the bulk vertices. Thus, we omit the boundary term in \eqref{eq:Hlimit1} and redefine $H(\vec\clz,\vec\clv,\vec\clc,\vec\clb)$ as
\begin{equation}\label{eq:Hlimit}
\begin{aligned}
&H(\vec\clz,\vec\clv,\vec\clc,\vec\clb)=\sum_v\Bigg\{\frac{1 }{4\pi \beta^2\Delta\Pl^2}|\clv(v)|\sin\left[4\pi\beta\sqrt{\Delta}\Pl\clb(v)\right]\sin\left[\frac{4\pi\beta\sqrt{\Delta}\Pl}{\clv(v)}\clc(v)+4\pi\beta\sqrt{\Delta}\Pl\clb(v)  \right]+\\
&\frac{1}{8\pi\beta^2\Delta\Pl^2}|\clv(v)|\sin^2(4\pi\beta\sqrt{\Delta}\Pl\clb(v))+\frac{|\clv(v)|}{8\pi\clz(v)^2}+ \frac{\pi}{2}\frac{1}{|\clv(v)|}\left[\sgn(\clz(v+1))\clz(v+1)^2-\sgn(\clz(v))\clz(v)^2\right]^2\Bigg\}.
\end{aligned}
\end{equation}

\subsection{The equations of motion}
According to the above discussions, we get an effective Hamiltonian as
\begin{equation}\label{eq:Heffdis}
\begin{aligned}
\efh(\vec\clz,\vec\clv,\vec\clc,\vec\clb)=-\frac{1}{G}H(\vec\clz,\vec\clv,\vec\clc,\vec\clb).
\end{aligned}
\end{equation}
Due to the Poisson brackets \eqref{eq:poissondis}, the EOMs are
\begin{equation}\label{eq:EOM}
\begin{aligned}
\frac{\dd \clv(v)}{\dd t}=&\frac{ \clv(v) \left(2 \sin \left[2 \fb(v) +\fc(v)\right]+\sin \left[2\fb(v)\right]\right)}{2 \beta  \sqrt{\Delta }\Pl}\\
\frac{\dd \clb(v)}{\dd t}=&-\frac{\sin(\fb(v))}{\beta\sqrt{\Delta}\Pl} \Bigg[\frac{2 \sin \left[\fb(v)+\fc(v)\right]+\sin \left[\fb(v)\right] }{8 \pi  \beta  \sqrt{\Delta}  \Pl}-\frac{\clc(v)}{\clv(v)} \cos \left[\fb(v)+\fc(v)\right] \Bigg]\\
&-\frac{1}{8\pi\clz(v)^2}+\frac{ \pi}{2} \frac{1}{ \clv(v)^2} \left(\clz(v)^2-\clz(v+1)^2\right)^2\\
\frac{\dd \clz(v)}{\dd t}=&\frac{ \clz(v) \sin \left[\fb(v)\right] \cos \left[\fb(v)+\fc(v)\right]}{\beta  \sqrt{\Delta } \Pl}\\
\frac{\dd \clc(v)}{\dd t}=&\frac{\clv(v)}{4 \pi  \clz(v)^2}+2\pi\clz(v)^2  \left(\frac{\clz(v+1)^2-\clz(v)^2}{\clv(v)}-\frac{\clz(v)^2-\clz(v-1)^2}{\clv(v-1)}\right)
\end{aligned}
\end{equation}
where
\begin{equation}
\begin{aligned}
\fb(v)=4\pi\beta\sqrt{\Delta}\Pl \clb(v),\ \fc(v)=4\pi\beta\sqrt{\Delta}\Pl \frac{\clc(v)}{\clv(v)}.
\end{aligned}
\end{equation}
In \eqref{eq:EOM}, we have assumed that $\zeta(v)>0$ without loss of generality, since, as seen below, $\sgn(\zeta(v))$ is kept along the concerning dynamical trajectary. 

In order to discuss the continuous limit of the EOMs \eqref{eq:EOM}, we introduce the continuous variables $\conv(x)$, $\conb(x)$, $\conz(x)$ and $\conc(x)$, which are related to the classical variables $\clv(v)$, $\clb(v)$, $\clz(v)$ and $\clc(v)$ via
\begin{equation}\label{eq:relationdiscon}
\begin{aligned}
&\clv(v)=\int_{e(v)}\dd x\, \conv(x),\ &&\clb(v)=\conb({\rm mid}_{e(v)})\\
&\clz(v)= \conz({\rm mid}_{e(v)}),\ &&\clc(v)=\int_{e(v)}\dd x\,\conc(x).
\end{aligned}
\end{equation}
With the continuous variables, \eqref{eq:relationnewold} is rewritten as
\begin{equation}\label{eq:relationnewoldp}
\begin{aligned}
&K_x(x)=\frac{\conc(x)+\conb(x)\conv(x)}{2 \conz(x)},\ &&E^x(x)=\sgn(\conz(x))\conz(x)^2,\\
&K_\varphi(x)=4\pi\conz(x)\conb(x),\ &&E^\varphi(x)=\frac{\conv(x)}{4\pi\conz(x)}. 
\end{aligned}
\end{equation}
Moreover, substituting the continuous variables into the EOMs \eqref{eq:EOM} and considering the continuous limit of the lattice $\gamma$, we get
\begin{equation}\label{eq:EOMcon}
\begin{aligned}
\partial_t \conv&=\frac{ \conv \left(2 \sin \left[2 \tilde\fb +\tilde\fc\right]+\sin \left[2\tilde\fb\right]\right)}{2 \beta  \sqrt{\Delta }\Pl}\\
\partial_t \conb&=-\frac{\sin(\tilde\fb)}{\beta\sqrt{\Delta}\Pl} \Bigg(\frac{2 \sin \left[\tilde\fb+\tilde\fc\right]+\sin \left[\tilde\fb\right] }{8 \pi  \beta  \sqrt{\Delta}  \Pl}-\frac{\conc}{\conv} \cos \left[\tilde\fb+\tilde\fc\right] \Bigg)-\frac{1}{8\pi\conz{}^2}+\frac{ \pi}{2} \frac{ \big(\partial_x(\conz{}^2)\big)^2}{ \conv{}^2}\\
\partial_t \conz&=\frac{ \conz\sin \left[\tilde\fb\right] \cos \left[\tilde\fb+\tilde\fc\right]}{\beta  \sqrt{\Delta } \Pl}\\
\partial_t \conc&=\frac{\conv}{4\pi  \conz{}^2}+2 \pi  \conz{}^2 \partial_x\left(\frac{1}{\conv}\partial_x(\conz{}^2) \right)
\end{aligned}
\end{equation}
where
\be
\tilde\fb(x)=4\pi\beta\sqrt{\Delta}\Pl \conb(x),\ \tilde\fc(x)=4\pi\beta\sqrt{\Delta}\Pl \frac{\conc(x)}{\conv(x)}.
\ee
The same EOMs as \eqref{eq:EOMcon} can be obtained if one consider a system of the phase space of $(\conv(x),\conb(x),\conz(x),\conc(x))$, with the non-vanishing Poisson brackets
\begin{equation}\label{eq:poissoncon}
\begin{aligned}
\{\conb(x),\conv(y)\}=G\delta(x,y),\ \{\conc(x),\conz(y)\}=G\conz(x)\delta(x,y)
\end{aligned}
\end{equation}
and the Hamiltonian
\begin{equation}\label{eq:Hlimitcon}
\begin{aligned}
&\tilde H_{\rm eff}=-\frac{1}{G}\int_{\mathbb R}\Bigg\{\frac{1 }{4\pi \beta^2\Delta\Pl^2}|\conv(x)|\sin\left[4\pi\beta\sqrt{\Delta}\Pl\conb(x)\right]\sin\left[\frac{4\pi\beta\sqrt{\Delta}\Pl}{\conv(x)}\conc(x)+4\pi\beta\sqrt{\Delta}\Pl\conb(x)  \right]+\\
&\frac{1}{8\pi\beta^2\Delta\Pl^2}|\conv(x)|\sin^2(4\pi\beta\sqrt{\Delta}\Pl\conb(x))+\frac{|\conv(x)|}{8\pi\conz(x)^2}+ \frac{\pi}{2}\frac{1}{|\conv(x)|}\left(\partial_x\left(\sgn(\conz(x))\conz(x)^2\right)\right)^2\Bigg\}.
\end{aligned}
\end{equation}
In this sense, this system described by \eqref{eq:poissoncon} and \eqref{eq:Hlimitcon} are the continuous limit of effective dynamics based on the lattice $\gamma$ and encoded in \eqref{eq:poissondis} and \eqref{eq:Heffdis}. {Since \eqref{eq:Hlimitcon} is the same as the effective Hamiltonian used in \cite{Han:2020uhb} if \eqref{eq:relationnewoldp} is substituted, \eqref{eq:EOMcon} is the same as the EOMs used in \cite{Han:2020uhb}.
}

\subsection{A solution to the EOMs}
The continuous effective Hamiltonian \eqref{eq:Hlimitcon} returns to the classical Hamiltonian \eqref{eq:clH} in the low curvature region where $\conb(x)\ll 1$ and $\conc(c)/\conv(x)\ll 1$. As is known in \cite{Han:2020uhb}, solving the EOMs generated by the classical Hamiltonian  \eqref{eq:clH} gives the Schwarzschild metric in Lem\^aitre coordinates. For this solution, all of the dynamical variables depend only on $x-t$, which represents homogeneity of the interior and the static feature of the exterior of the Schwarzschild solution. 
In this section, we are also concerned about the solutions depending only one $x-t$, that is, solutions taking the form
\begin{equation}\label{eq:ansatz}
\conv(x,t)=\conv(x-t),\ \conb(x,t)=\conb(x-t),\ \conz(x,t)=\conz(x-t),\ \conc(x,t)=\conc(x-t).
\end{equation}

Before proceeding further, it is helpful to investigate constants of motion in the effective dynamics. As mentioned in \eqref{eq:vx}, $\cv(x)$ is a conversed charge in the classical dynamics. Fortunately, this feature is kept in the effective dynamics, because of $0=\{\cv(x),\tilde H_{\rm eff}\}.$
In terms of the continuous variables, $\cv(x)$ is
\begin{equation}
\cv(x)=\frac{\conc(x)\partial_x\conz(x)}{\conz(x)}-\conv(x)\partial_x\conb(x)
\end{equation}
Since we are concern on the effective solutions which can recover the Schwarzschild metric in the low curvature region,
it has
\begin{equation}\label{eq:vanishvneff}
\cv(x)=\frac{\conc(x)\partial_x\conz(x)}{\conz(x)}-\conv(x)\partial_x\conb(x)=0
\end{equation}  
with recalling the statements below \eqref{eq:vx}.

Substitute the ansatz \eqref{eq:ansatz} in to the EOMs \eqref{eq:EOMcon} and \eqref{eq:vanishvneff}. With denoting $y=x-t$, we have
\begin{eqnarray}
-\frac{\dd \conv(y)}{\dd y}&=&\frac{ \conv(y) \left(2 \sin \left[2 \tilde\fb(y) +\tilde\fc(y)\right]+\sin \left[2\tilde\fb(y)\right]\right)}{4 \beta  \sqrt{\Delta }\Pl}\label{eq:EOMansatz1}\\
-\frac{\dd  \conb(y)}{\dd y}&=&-\frac{\sin(\tilde\fb(y))}{\beta\sqrt{\Delta}\Pl} \Bigg(\frac{2 \sin \left[\tilde\fb(y)+\tilde\fc(y)\right]+\sin \left[\tilde\fb(y)\right] }{8 \pi  \beta  \sqrt{\Delta}  \Pl}-\frac{\conc(y)}{\conv(y)} \cos \left[\tilde\fb(y)+\tilde\fc(y)\right] \Bigg)\label{eq:EOMansatz2}\\
&&-\frac{1}{8\pi\conz(y){}^2}+\frac{ \pi}{2} \frac{ 1}{ \conv(y){}^2}\Big(\frac{\dd \conz(y){}^2}{\dd y}\Big)^2\nonumber\\
-\frac{\dd  \conz(y)}{\dd y}&=&\frac{ \conz(y)\sin \left[\tilde\fb(y)\right] \cos \left[\tilde\fb(y)+\tilde\fc(y)\right]}{\beta  \sqrt{\Delta } \Pl}\label{eq:EOMansatz3}\\
-\frac{\dd\conc(y)}{\dd y}&=&\frac{\conv(y)}{4 \pi  \conz(y){}^2}+2\pi \conz(y){}^2 \frac{\dd}{\dd y}\left(\frac{1}{\conv(y)}\frac{\dd\conz(y){}^2}{\dd y} \right)\label{eq:EOMansatz4}
\end{eqnarray}
and
\begin{equation}\label{eq:diffansatz}
\frac{\conc(y)}{\conz(y)}\frac{\dd \conz(y)}{\dd y}-\conv(y)\frac{\dd\conb(y)}{\dd y}=0
\end{equation}
To get the solutions, the equations \eqref{eq:EOMansatz1}, \eqref{eq:EOMansatz3}, \eqref{eq:EOMansatz4} and \eqref{eq:diffansatz} are chosen. A set of numerical results is shown in Fig. \ref{fig:solution}. 
Since \eqref{eq:EOMansatz2} is used to get the solution, we can use it to check the accuracy of our solutions. Substituting the numerical solutions achieved from  \eqref{eq:EOMansatz1}, \eqref{eq:EOMansatz3}, \eqref{eq:EOMansatz4} and \eqref{eq:diffansatz} into \eqref{eq:EOMansatz2}, we have the residuals plotted in Fig. \ref{fig:error}. 
According to the numerical results, for $y=x-t\to-\infty$, the variables behaves as
\begin{equation}
\begin{aligned}
\conv(x-t)=e^{\fa+\fb|x-t|\Pl^{-1}},\ \conb(x-t)=\fc,\ \conz(x-t)=r_0,\ \conc(x-t)=\fd
\end{aligned}
\end{equation} 
where $\fa$, $\fb$, $\fc$, $r_0$ and $\fd$ are constants. Thus, at $y\to -\infty$ is the spacetime is diffeomorphism to Nariai geometry taking $\mathrm{dS}_2\times S^2$ metric
\begin{equation}
\dd s^2=-\dd t^2+\left(\frac{\conv(x,t)}{\conz{}^2}\right)^2\dd x{}^2+\conz{}^2\dd\Omega^2,
\end{equation}
where the metric in terms of the variables $\conv$ and $\conz$ is given in \cite{Chiou:2012pg,Han:2020uhb}. This result is the same as the results given in \cite{Han:2020uhb}, which is not surprising because the equations  \eqref{eq:EOMansatz1}, \eqref{eq:EOMansatz2}, \eqref{eq:EOMansatz3} and \eqref{eq:EOMansatz4} are the same as that used in \cite{Han:2020uhb}. Thus one can refer to \cite{Han:2020uhb} for more details on this solution. According to the analysis there, both the area of $\mathbb S^2$ and the $\mathrm{dS}$ radius of the Nariai geometry obtained in this model are of quantum size.
In other words, the current model predicts a quantum final fate of Schwarzschild BH. However, it should be emphasized that this is not credible always, which will be shown below by discussing the scope of the continuous effective descriptions.

\begin{figure}
\centering
\includegraphics[width=0.8\textwidth]{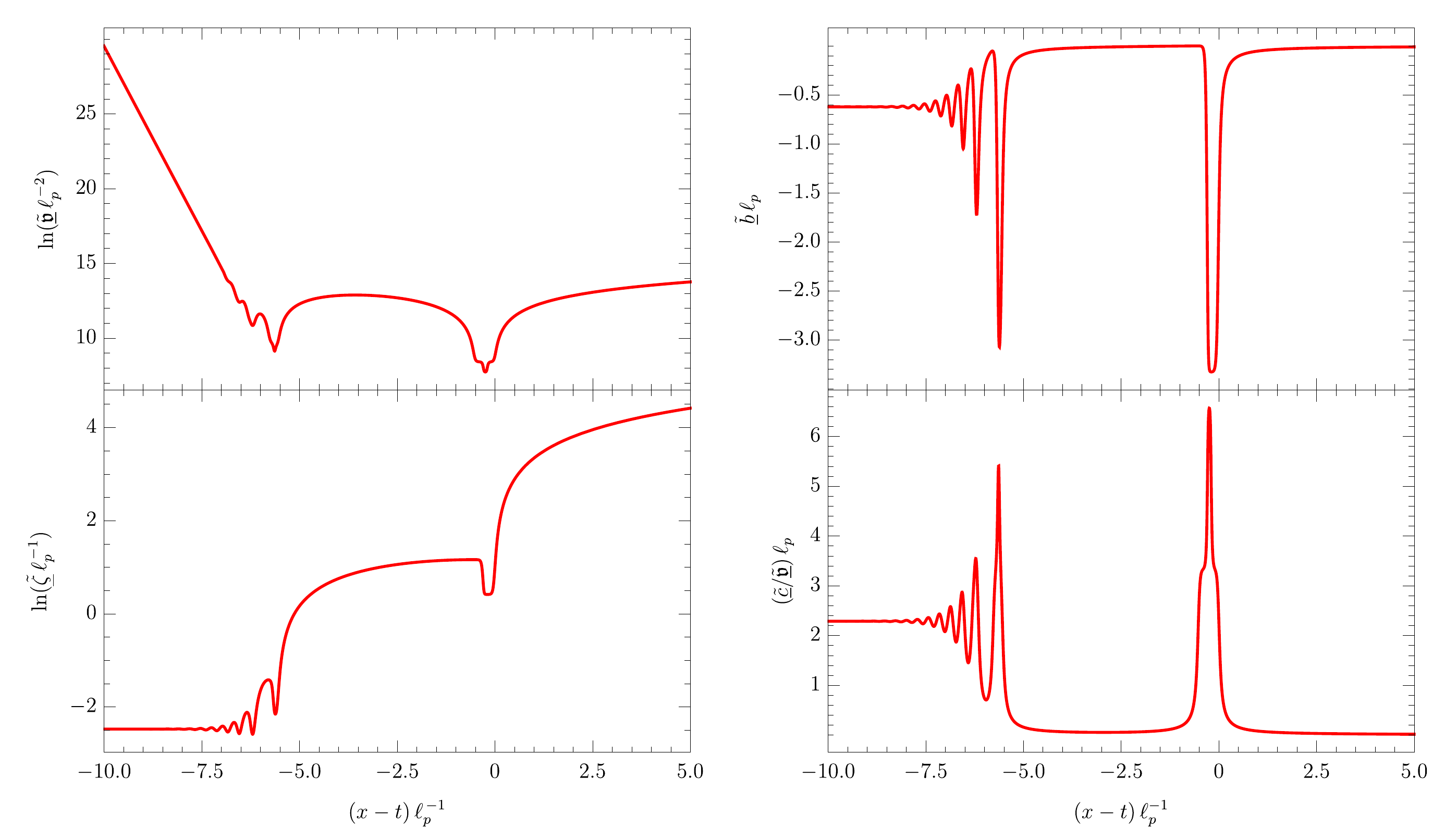}
\caption{Solutions to \eqref{eq:EOMansatz1}, \eqref{eq:EOMansatz3}, \eqref{eq:EOMansatz4} and \eqref{eq:diffansatz}. The initial data are chosen such that  $\conv(y_0)=6\pi F_0 y_0$, $\conb(y_0)=-\frac{1}{6\pi y_0}$, $\conz(y_0)=\left(\frac{3}{2} \sqrt{F_0} y_0\right)^{2/3}$, $\conc(y_0)=\frac{3F_0}{2}$ and $\conz'(y_0)=\frac{2}{3}\left(\frac{3}{2} \sqrt{F_0} y_0\right)^{-1/3}$ with $y_0=10^5\Pl$ and $F_0=10^4\Pl$. The initial data are choose by considering the Schwarzschild solution with $2GM=F_0$  in Lem\^aitre coordinate at $x-t=y_0$.  The parameters are set to be $\Delta=0.1$ and $\beta=0.2375$. }
\label{fig:solution}
\end{figure}

\begin{figure}[h]
\centering
\includegraphics[width=0.5\textwidth]{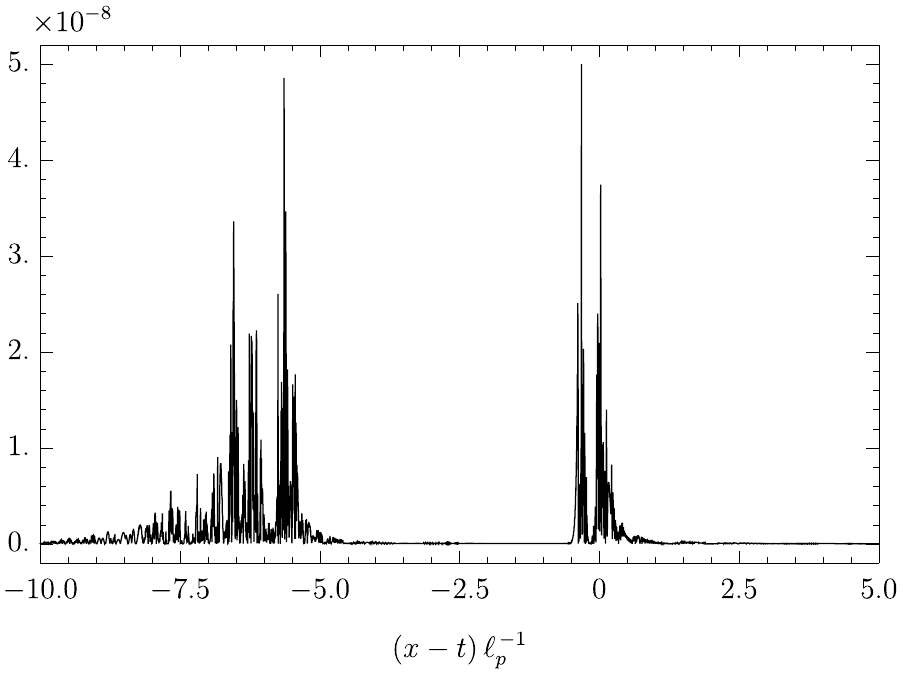}
\caption{The residuals of \eqref{eq:EOMansatz2} by substituting the numerical solutions shown in Fig. \ref{fig:solution}}
\label{fig:error}
\end{figure}

Let us discuss the scope of continuous effective differential equations \eqref{eq:EOMansatz1}--\eqref{eq:EOMansatz4}. At first, the derivation from \eqref{eq:Horigion} to \eqref{eq:Hlimit} requires
\begin{equation}
\clv(v)\gg 4\pi\beta\sqrt{\Delta}\Pl^3,\ \forall v\in V(\gamma).
\end{equation}
According to \eqref{eq:relationdiscon}, this means
\begin{equation}\label{eq:criterion}
\int_{e(v)}\dd x\,\conv(x)\gg 4\pi\beta\sqrt{\Delta}\Pl^3, \forall v\in V(\gamma).
\end{equation}
Let $\conv{}_{\rm min}$ be the smallest value of $\conv(x,t)$. Since $\int_{e(v)}\dd x\,\conv(x)>\delta x\, \conv{}_{\rm min}$ where $\delta x $ denotes the coordinate length of $e(v)$, it is concluded that \eqref{eq:criterion} can be satisfied provided that
\begin{equation}\label{eq:criterion1}
\delta x\, \conv{}_{\rm min}\gg 4\pi\beta\sqrt{\Delta}\,\Pl^3.
\end{equation}

\begin{figure}[h]
\centering
\includegraphics[width=0.5\textwidth]{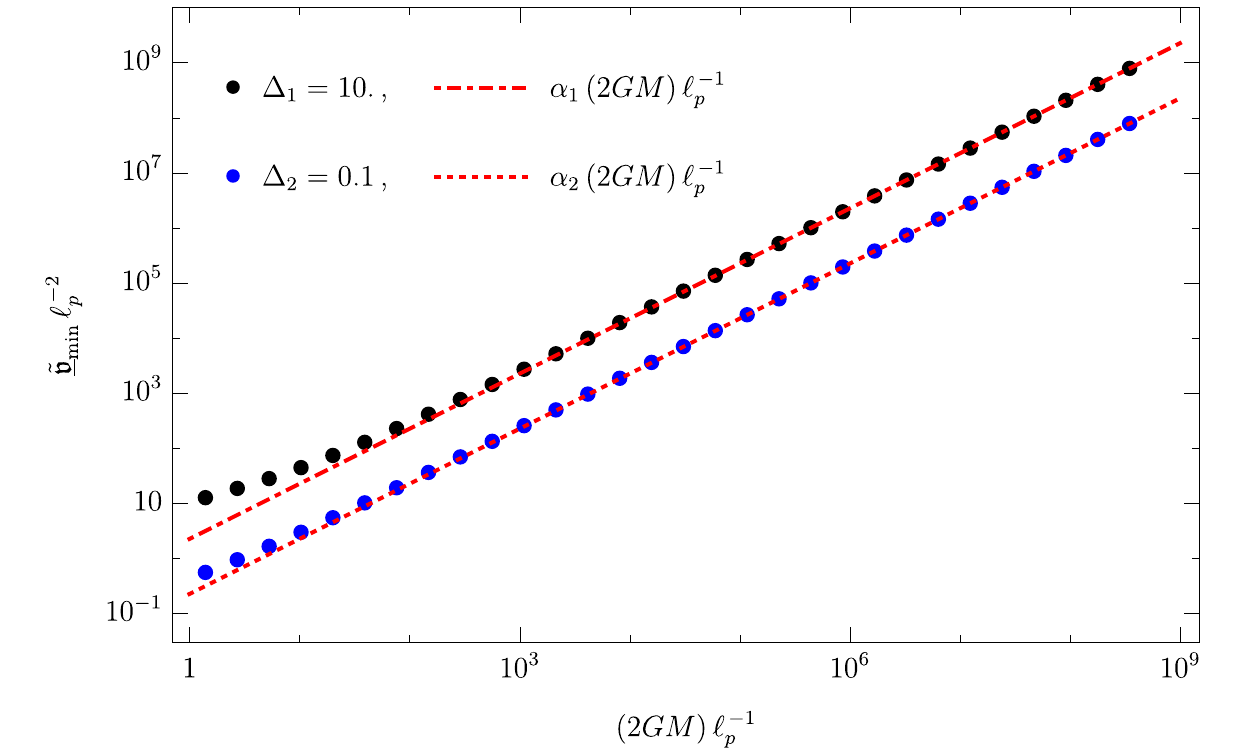}
\caption{Dependence of $\conv_{\rm min}$ on $M$ for various values of $\Delta$. As shown in the figure, for $M\gg 1$, $\conv_{\rm min}$ can be well approximated by \eqref{eq:approximation0}. In this plot, the values of $\alpha_i$, corresponding to $\Delta_i$, with $i=1,\ 2$ are computed by \eqref{eq:alpha0}. The parameters are chosen as $\beta=0.2375$. 
}
\label{fig:vmin}
\end{figure}

The solutions to \eqref{eq:EOMansatz1}--\eqref{eq:EOMansatz4} is determined by the initial data.  
Meanwhile, the initial data are chosen such that in the low curvature region the effective dynamics returns to the Schwarzschild metric which depends only on its mass $M$. As a consequence, the value  $\conv{}_{\rm min}$, as the minimal value of the solution, is related to the mass of the Schwarzschild spacetime to set the initial data. According to our numerical computation shown in Fig. \ref{fig:vmin}, the corelation between $\conv_{\rm min}$ and $M$, for $M\gg 1$, can be well approximated by a formula of the form
\begin{equation}\label{eq:approximation0}
\conv_{\rm min}\cong \alpha (2GM)\Pl.
\end{equation}
where $\alpha$ is some coefficient depends only on $\beta\sqrt{\Delta}\Pl$, i.e. $\alpha=\alpha(\beta\sqrt{\Delta}\Pl)$, because the coefficients in \eqref{eq:EOMansatz1}--\eqref{eq:EOMansatz4} depends on $\beta\sqrt{\Delta}\Pl$ merely. The dependence of $\alpha$ on $\beta\sqrt{\Delta}\Pl$ can be investigated numerically, which, as shown in Fig. \ref{fig:alpha0}, is
\begin{equation}\label{eq:alpha0}
\alpha(\beta\sqrt{\Delta}\,\Pl)\cong\alpha_0 \beta\sqrt{\Delta}
\end{equation}
with $\alpha_0\cong3.0000$. Therefore, for $M\gg 1$, $\conv_{\rm min}$ is approximated by
\begin{equation}\label{eq:approximation}
\conv_{\rm min}\cong \alpha_0\beta\sqrt{\Delta}(2GM) \Pl.
\end{equation}
\begin{figure}
\centering
\includegraphics[width=0.5\textwidth]{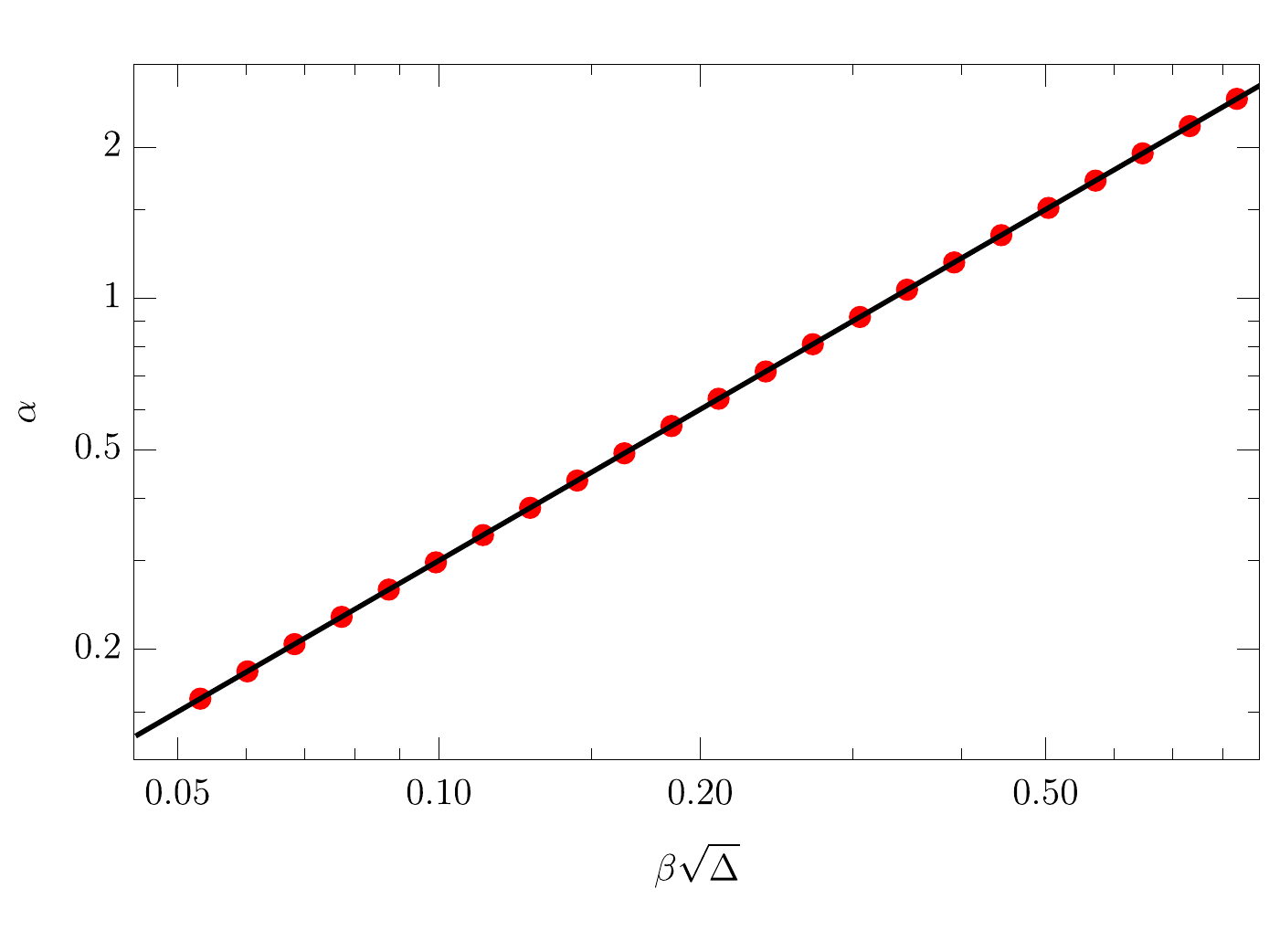}
\caption{Plot of the values $\alpha$, introduced in \eqref{eq:approximation}, depending on $\beta\sqrt{\Delta}\,\Pl$. As shown in the figure, $\alpha$ depends on $\beta\sqrt{\Delta}$ linearly. More explicitly, $\alpha=\alpha_0\beta\sqrt{\Delta}$ with $\alpha_0=3.0000$ according to the numerical results shown here.  }
\label{fig:alpha0}
\end{figure}

Substituting \eqref{eq:approximation} into $\eqref{eq:criterion1}$, we finally obtain that 
\begin{equation}\label{eq:scope}
\frac{2GM}{\Pl}\gg \frac{4\pi}{\alpha_0}\,\frac{\Pl}{\delta x}
\end{equation}
which presents the scope of the continuous effective EOMs. It should be emphasized that $\delta x$ is the coordinate length of each edge in a specific coordinate such that the effective metric depends only on $x-t$ and, takes the limit of Schwarzschild metric in Lem\^aitre coordinate in the low curvature region. Moreover, the above discussion assumed that the lattice $\gamma$ is equidistant in this specific coordinate.

According to \eqref{eq:scope},  $\delta x$ should take large values to enlarge the scope of the continuous effective EOMs. However, since we used differentials to approximate
the differences in \eqref{eq:EOM}, a tension arises that $\delta x$ cannot be too large. To see more explicitly how this tension limits $\delta x$, one notices that the approximations of differences by differentials in \eqref{eq:EOM} are 
\begin{equation}
\begin{aligned}
\frac{1}{\clv(v)}\left(\clz(v+1)^2-\zeta(v)^2\right)&\cong \left.\frac{1}{\conv(x)}\partial_x(\conz(x)^2)\right|_{x={\rm mid}_{e(v)}}\\
 \left(\frac{\clz(v+1)^2-\clz(v)^2}{\clv(v)}-\frac{\clz(v)^2-\clz(v-1)^2}{\clv(v-1)}\right)&\cong \left.\partial_x\left(\frac{1}{\conv(x)}\partial_x(\conz(x){}^2) \right)\right|_{x={\rm mid}_{e(v)}}.
\end{aligned}
\end{equation} 
The last equations omit the dependence of all variables on $t$ for simplicity. That is, we are discussing the validity these two equations for a fixed $t$. 
 Consider the first equation as an example.
According to \eqref{eq:relationdiscon}, one has
\begin{equation}\label{eq:taylorseries}
\begin{aligned}
\frac{1}{\clv(v)}\left(\clz(v+1)^2-\zeta(v)^2\right)=\left(\frac{\partial_x(\conz{}^2)({\rm mid}_{e(v)})}{\conv({\rm mid}_{e(v)})}+\frac{1}{2}\frac{\partial_x^2(\conz{}^{2})(c')\,\delta x }{\conv({\rm mid}_{e(v)}) } \right)\frac{1}{1+\frac{\partial_x\conv{}(c)\delta x}{2\conv({\rm mid}_{e(v)})}}
\end{aligned}
\end{equation}
where $c$ and $c'$ are some points in $e(v)$ and the mean-value form of the remainder of the Taylor series is applied. As mentioned above, in the region $x-t\gg 1$, the solutions to the effective continuous EOMs converge to the Schwarzschild metric in Lem\^aitre coordinate, i.e.
\begin{equation}
\begin{aligned}
\conz(x,t)=&\left(\frac{3}{2}\sqrt{F}(x-t)\right)^{2/3}\\
\conv(x,t)=&6\pi F (x-t)
\end{aligned}
\end{equation}
with some constant $F$. Substituting these expressions into \eqref{eq:taylorseries}, we have
\begin{equation}
\frac{1}{\clv(v)}\left(\clz(v+1)^2-\zeta(v)^2\right)=\frac{\partial_x(\conz{}^2)({\rm mid}_{e(v)})}{\conv({\rm mid}_{e(v)})}+o(\delta x/(x-t)).
\end{equation}
Thus in the region $x-t\gg 1$, the differentials can well approximate the differences if $\delta x\sim 1$. Moreover, because $\conz(x,t)$ keeps constant in the region $x-t\ll -1$, the differentials can well approximate the differences in this region for any value of $\delta x$. As a summary, the approximation of the differences by differentials is well behaved in the region $|x-t|\gg 1$. Thus, we only need to check the validity  of this approximation in the region of $|x-t|\sim 1$.  To do this, we substitute \eqref{eq:relationdiscon} into \eqref{eq:EOM} and numerically compute the following dimensionless variables
\begin{equation}\label{eq:diserrs}
\begin{aligned}
\text{err}_{\clv}(v,t)&:=-\frac{\Pl}{ \clv(v,t) }\frac{\dd \clv(v,t)}{\dd t}+\frac{\left(2 \sin \left[2 \fb(v,t) +\fc(v,t)\right]+\sin \left[2\fb(v,t)\right]\right)}{2 \beta  \sqrt{\Delta }},\\
\text{err}_{\clb}(v,t)&:=-\Pl^2\frac{\dd \clb(v,t)}{\dd t}-\frac{\sin(\fb(v,t))}{\beta\sqrt{\Delta}} \Bigg[\frac{2 \sin \left[\fb(v,t)+\fc(v,t)\right]+\sin \left[\fb(v,t)\right] }{8 \pi  \beta  \sqrt{\Delta} }-\Pl\frac{\clc(v,t)}{\clv(v,t)} \cos \left[\fb(v,t)+\fc(v,t)\right] \Bigg]\\
&-\frac{\Pl^2}{8\pi\clz(v,t)^2}+\frac{ \pi}{2} \frac{\Pl^2}{ \clv(v,t)^2} \left(\clz(v,t)^2-\clz(v+1,t)^2\right)^2,\\
\text{err}_{\clz}(v,t)&:=-\frac{\Pl}{\clz(v,t)}\frac{\dd \clz(v,t)}{\dd t}+\frac{ \sin \left[\fb(v,t)\right] \cos \left[\fb(v,t)+\fc(v,t)\right]}{\beta  \sqrt{\Delta }},\\
\text{err}_{\clc}(v,t)&:=-\frac{\Pl^2}{\clv(v,t)}\frac{\dd \clc(v,t)}{\dd t}+\frac{\Pl^2}{4 \pi  \clz(v,t)^2}+2\pi\Pl^2\frac{\clz(v,t)^2}{\clv(v,t)}  \left(\frac{\clz(v+1,t)^2-\clz(v,t)^2}{\clv(v,t)}-\frac{\clz(v,t)^2-\clz(v-1,t)^2}{\clv(v-1,t)}\right).
\end{aligned}
\end{equation}
Fig \ref{fig:diserrsonM} plots the values of $\max|\text{err}_{\clv}|$, $\max|\text{err}_{\clb}|$, $\max|\text{err}_{\clz}|$ and $\max|\text{err}_{\clc}|$ which are defined by
\begin{equation}
\begin{aligned}
\max|\text{err}_{\clv}|:=&\max_{t\in\mathbb R,v\in V(\gamma)}|\text{err}_{\clv}(v,t)|,\\
\max|\text{err}_{\clb}|:=&\max_{t\in\mathbb R,v\in V(\gamma)}|\text{err}_{\clb}(v,t)|,\\
\max|\text{err}_{\clz}|:=&\max_{t\in\mathbb R,v\in V(\gamma)}|\text{err}_{\clz}(v,t)|,\\
\max|\text{err}_{\clc}|:=&\max_{t\in\mathbb R,v\in V(\gamma)}|\text{err}_{\clc}(v,t)|.
\end{aligned}
\end{equation} 
According to \eqref{eq:scope}, the minimal values of $2GM\Pl^{-1}$ for a fixed $\delta x$ is $(2GM\Pl^{-1})_{\rm min}\cong 4.19/ (\delta x\Pl^{-1})$, which is $4.19\times 10^{-4}$ for the case plotted in 
Fig. \ref{fig:diserrsonM}. Thus, the inequality of \eqref{eq:scope} is satisfied for $2GM\geq 10^6\Pl$.  Moreover, as shown in Fig. \ref{fig:diserrsonM}, the values of $\max|\text{err}_{\clv}|$, $\max|\text{err}_{\clb}|$, $\max|\text{err}_{\clz}|$ and $\max|\text{err}_{\clc}|$ are tiny. Thus the solutions to the differential equations \eqref{eq:EOMcon} approximate well the solutions to the difference EOMs \eqref{eq:EOM} for large $2GM$. Note that for large $2GM$, $\delta x$ can be chosen small. However, for small $2GM$, this conclusion is no longer valid due to the large value of the allowed $\delta x$. This can be seen intuitively from the numerical results.  
 Fig. \ref{fig:diserrs2} plots the values of $\max|\text{err}_{\clv}|$, $\max|\text{err}_{\clb}|$, $\max|\text{err}_{\clz}|$ and $\max|\text{err}_{\clc}|$ depending on $\delta x$ for $2GM=100\Pl$. For $2GM=100\Pl$, the allowed values of $\delta x$ by \eqref{eq:scope} are those satisfying $\delta x\,\Pl^{-1}\gg 4.19\times 10^{-2}$ which are indicated by the red lines in Fig. \ref{fig:diserrs2}. However, as shown in  Fig. \ref{fig:diserrs2}, the values of  $\max|\text{err}_{\clv}|$, $\max|\text{err}_{\clb}|$, $\max|\text{err}_{\clz}|$ and $\max|\text{err}_{\clc}|$  are no longer small for the allowed $\delta x$. Thus the solutions to the differential equations \eqref{eq:EOMcon} no longer well approximate the solutions to the difference EOMs \eqref{eq:EOM}. 
\begin{figure}
\centering
\includegraphics[width=0.9\textwidth]{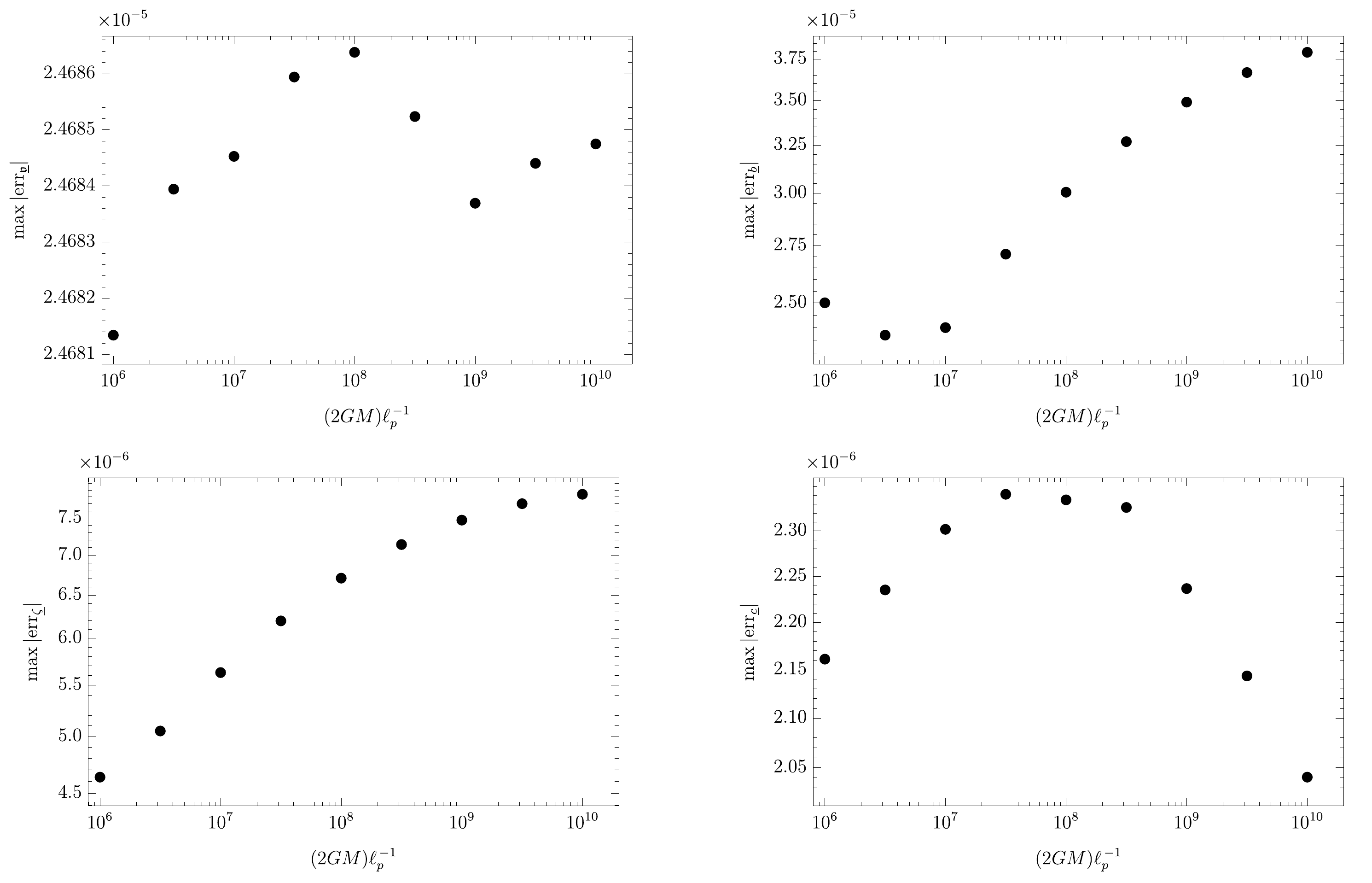}
\caption{Plots of the values of $\max|\text{err}_{\clv}|$, $\max|\text{err}_{\clb}|$, $\max|\text{err}_{\clz}|$ and $\max|\text{err}_{\clc}|$ depending on $2GM$ for $\delta x=10^{-4}\Pl$. By \eqref{eq:scope}, the minimal values of $(2GM)\Pl^{-1}$ in this case is $4\pi/(\alpha_0\delta x\Pl^{-1} )\approx 4.19\times 10^{4}$. 
}
\label{fig:diserrsonM}
\end{figure}
\begin{figure}
\centering
\includegraphics[width=0.9\textwidth]{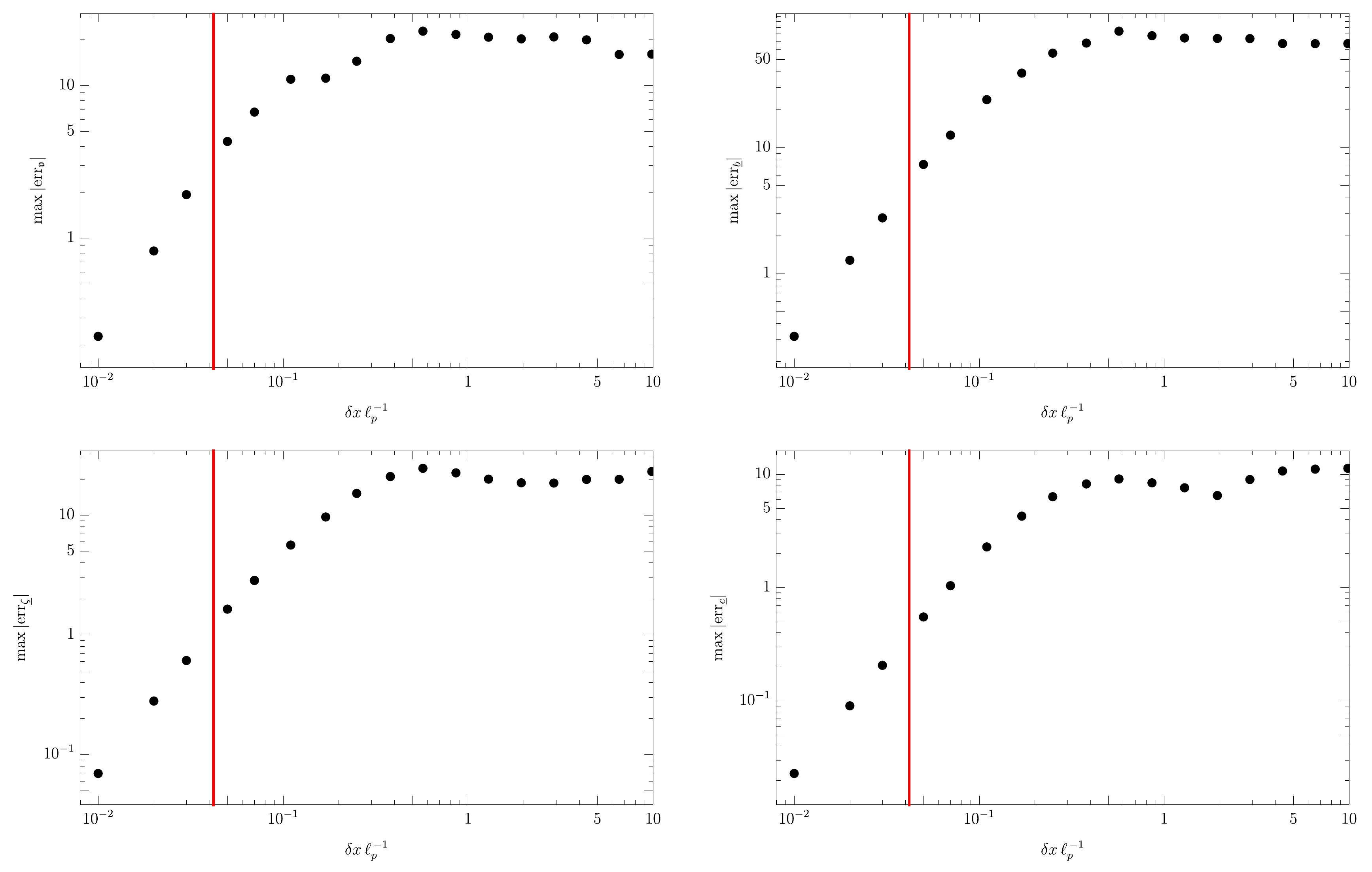}
\caption{ The values of $\max|\text{err}_{\clv}|$, $\max|\text{err}_{\clb}|$, $\max|\text{err}_{\clz}|$ and $\max|\text{err}_{\clc}|$ depending on $\delta x$ for $2GM=100\Pl$ (black dots). $\delta x\,\Pl^{-1}=4\pi/(\alpha_0(2GM)\Pl^{-1})$ is plotted by the red lines which indicate the allowed values of $\delta x\,\Pl^{-1}$ by \eqref{eq:scope}. }
\label{fig:diserrs2}
\end{figure}

By this analysis, the discrete EOMs  \eqref{eq:EOM} {have correct classical limit}
only if the inequality \eqref{eq:scope} is satisfied and, simultaneously $\delta x$ is so small that solutions to \eqref{eq:EOM} are compatible with that of the continuous effective dynamics which ensures that the classical Schwarzschild solutions are recovered asymptotically.
However, these two conditions lead to opposite tendencies for the values of $\delta x$. According to \eqref{eq:scope}, $\delta x$ should be large so that \eqref{eq:EOM} can be valid for a large scope, which is contradictory to the second requirement that $\delta x$ should be small such that the continuous EOMs can well approximate the discrete ones. This tension prevent the value of $\delta x$ to vanishing, because for vanishing $\delta x$, the effective EOMs are no longer valid for any $2GM$. However, on the other hand, this tension requires $\delta x$ to be small enough, because otherwise the continuous EOMs cannot give well-approximated solutions to the discrete EOMs. This argument implies that $\delta x$ should be some finitely small value. Recalling that $\delta x$ is the length of the edges in $\gamma$, we conclude that the current model does not allow $\gamma$ to shrink to a continuous line, even though the classical theory is presented based on a continuous line. Thus, the lattice structure underlying the current model is fundamental. However, this conclusion does not mean that the current lattice model never owns continuous limit. Since $\delta x$ is small, the discrete EOMs can be well approximated by the differential EOMs \eqref{eq:EOMcon}, which can be derived  from a continuous Hamiltonian \eqref{eq:Hlimitcon} together with the Poisson brackets \eqref{eq:poissoncon}. Thus, the continuous description of the lattice model is the model where the phase space, consisting of fields $(\conv(x),\conb(x),\conz(x),\conc(x))$ on $\mathbb R$, is endowed with the Poisson brackets \eqref{eq:poissoncon} and the Hamiltonian \eqref{eq:Hlimitcon}. This continuous description is usually referred to as the effective dynamics and has been studied in \cite{Han:2020uhb}. However, according to our analysis, this continuous description is only valid for black holes with macro masses satisfying \eqref{eq:scope} but not for small BHs. In other words, macro BHs and small BHs could behave differently if the LQG effects are considered.

\section{Conclusion and Outlook}\label{sec:eight}
This paper considers the loop quantum theory of spherically symmetric gravity coupled to Gaussian dusts where the Guassian dusts provide a material reference frame of the spaceand time to deparameterize gravity. Classically, the dynamics of this model can give the Schwarzschild solution in Lema\^itre coordinate. Thus this model provide an approach to study the spherically symmetric BH. As in the usual loop quantum theory, the present model is constructed base on some graph $\gamma$ which is an 1-dimensional lattice here. By using the $\bar\mu$-scheme to regularize the physical Hamiltonian, we obtain an Hamiltonian operator which governs the quantum dynamics. Then the quantum dynamics is studied by the path integral formulation and an effective action is obtained. With this effective action, an effective continuous description of the quantum lattice model is derived which is encoded in \eqref{eq:Hlimitcon}. Noted that \eqref{eq:Hlimitcon} is indeed the classical Hamiltonian with the holonomy correction which is usually used in literature like \cite{Han:2020uhb}. However, according to our analysis, this effective continuous description is valid only if \eqref{eq:scope} holds. In other works, this effective description, roughly speaking, is valid for macro BHs but not for small BHs. Therefore, it is natural to ask questions on the dynamics of small BHs. Moreover, it is also interesting to introduce BH evaporation to the current model and consider the issue on the final fate of the evaporation. These topics will be left for our further studying. 

\section*{Acknowledgements}
This work is benefited greatly from the numerous discussions with Muxin Han and Honguang Liu, and supported by the Polish Narodowe Centrum Nauki, Grant No. 2018/30/Q/ST2/00811.

\providecommand{\href}[2]{#2}\begingroup\raggedright\endgroup

\end{document}